\documentclass[11pt, oneside]{article} 
\usepackage[document]{ragged2e}
\usepackage{euscript, amsmath, amsthm, amsfonts, wasysym, verbatim, bbm, color, graphics, geometry, authblk, mathtools, MnSymbol, esint, hyperref}
\usepackage{placeins}
\usepackage[style=authoryear]{biblatex}
\newcommand{\R}{\mathbb{R}}

\newcommand{\N}{\mathbb{N}}

\newcommand{\prob}{\mathbb{P}}
\newcommand{\T}{\mathbb T}

\newcommand{\E}{\mathbb{E}}
\newcommand{\F}{\mathcal{F}}
\newcommand{\drift}{\boldsymbol{\Theta}}
\newcommand{\kdrift}{\hat{\boldsymbol{\Theta}}}
\newcommand{\kvar}{\boldsymbol{\Sigma}}
\newcommand{\cprice}{\mathfrak p}
\newcommand{\cwealth}{\mathfrak w}
\newcommand{\cdrift}{\boldsymbol{\vartheta}}
\newcommand{\cslow}{\mathfrak s}
\newcommand{\cfast}{\mathfrak f}
\newcommand{\cpricedrift}{\mathfrak z}
\newcommand{\quadraticstate}{\boldsymbol{\mathcal Q}}
\newcommand{\quadraticcoef}{\boldsymbol{Q}}
\newcommand{\linearcoef}{\boldsymbol{b}}
\newcommand{\quadraticshift}{f}
\newcommand{\K}{\mathcal K}
\newcommand{\augdrift}{\widetilde{\mu}}
\newcommand{\const}{\mathfrak c}
\newcommand{\volatilitystate}{\boldsymbol{v}}
\newcommand{\val}{\mathcal V}
\newcommand{\nprice}{\check{P}}
\newcommand{\Cdot}{\boldsymbol{\cdot}}
\newcommand*{\dd}{\mathop{}\!\mathrm{d}}

\newtheorem{theorem}{Theorem}[section]
\newtheorem{corollary}{Corollary}[section]
\newtheorem{lemma}{Lemma}[section]
\newtheorem{proposition}{Proposition}[section]
\theoremstyle{definition}
\newtheorem{definition}{Definition}[section]
\newtheorem{remark}{Remark}[section]

\renewcommand{\abstractname}{}
\renewenvironment{abstract}
 {
  \small
  \quotation
  \flushleft
  {\bfseries\noindent{\large\abstractname}\par\nobreak\smallskip}
 }
 {\endquotation}
 
\hypersetup{
    colorlinks=true,       
    linkcolor=blue,        
    citecolor=blue,       
    urlcolor=black,         
    filecolor=magenta      
}
 
\title{\bf Portfolio Optimization with Interacting Fast and Slow Latent Drift Factors}
\author{Dannin J. Eccles {\footnotesize AND} Roger Lee\\\vspace{-3mm}
{\footnotesize Department of Mathematics, University of Chicago, USA}}
\date{}
 
\begin{document}
 
\maketitle
 
%
%
 
\begin{abstract}
    {\large A}{BSTRACT.} $\;$We consider a class of partial-information portfolio optimization problems where a traded spread exhibits short-run momentum and long-run mean reversion driven by latent Gaussian factors. Using Kalman–Bucy filtering, we derive and verify optimal portfolios under logarithmic, exponential, and power utility over arbitrary finite horizons. In each case, optimal demand decomposes into a myopic bet on expected price changes and, under power and exponential utility, an additional intertemporal hedging component. At finite times, we show that myopic demand is generated by a causal affine functional of current price and a filter adaptive MACD signal, a difference of fast and slow exponential moving averages of price whose loadings follow the Kalman gain. As the gain converges, the myopic portfolio converges to one of two forms. When the separation between the decay rates of the fast and slow estimates dominates the feedback that learning creates between them, the limit is a function of current price, an MACD-type signal, and, under logarithmic and power utility, current wealth, and we characterize exactly when it is a function of current price and a conventional MACD indicator. When feedback dominates, the signal has damped oscillatory price-history weights.
\end{abstract}
 
%
%
 
\section{Introduction}
 
Moving-average rules such as the moving average convergence/divergence (MACD) indicator \parencite{Appel2005} are widely used in technical analysis, but they are specified by convention rather than derived from a model of prices. In this paper, we show that an MACD-type signal emerges endogenously in the optimal portfolio of a traded spread driven by two latent factors operating on different time scales. We then characterize exactly when this MACD-type signal is a conventional MACD indicator.

\medskip
A recurring lesson from the portfolio-choice literature is that predictability 
can operate over different time scales. In equity markets, evidence of long-horizon 
reversal is reported by \textcite{PoterbaSummers1988} and \textcite{FamaFrench1988}, 
while \textcite{JegadeeshTitman1993} show that strategies buying recent winners and 
selling recent losers earn positive returns over intermediate holding periods. For
 portfolio choice, these findings suggest that predictive forces acting at different 
 horizons need not point in the same direction: recent performance may favor one position 
 over a shorter horizon even when longer-run predictability points toward subsequent 
 reversal. Optimal allocation can therefore depend not only on the strength of 
 predictability, but also on the time scales over which these forces operate.
 
\medskip
\textcite{KoijenRodriguezSbuelz2009} show how this interaction shapes strategic
asset allocation when the two predictive states are observed. With a
recent-performance state and a persistent dividend-yield state, their calibrated
solution produces negative total hedging demand at intermediate horizons and a
nonmonotone relation between equity allocation and the investment horizon.
Thus, when predictive states operate on different time scales, their
interaction can materially alter both the level and the horizon dependence of
optimal portfolio exposure. The information available about those states
provides a second dimension of the allocation problem.
\textcite{AltayColaneriEksiAltayFlonner2026} make the persistent
mean-reversion state unobservable while retaining the return-history momentum
variable as an observable input, so that the investor learns about the hidden
state from prices while choosing the portfolio, and show how partial
observability changes the composition of optimal demand.
 
\medskip
Relative-value trading provides a natural setting in which predictive forces
can operate at different time scales. The central premise of many pairs-trading
and cointegration strategies is long-run convergence: related prices share a
persistent economic relation, and sufficiently large departures from that
relation create opportunities for positions that profit as the spread
subsequently narrows. \textcite{GatevGoetzmannRouwenhorst2006} document the
performance of such pairs strategies, while \textcite{LiuTimmermann2013}
study optimal investment when deviations between related prices generate
recurring or nonrecurring convergence opportunities. Spreads can also be
traded directly; \textcite{Cuny2006}, for example, studies futures contracts
written on the difference between futures prices at two delivery dates.
Convergence, however, need not occur monotonically or at a single time scale.
Temporary dislocations may persist before correcting, while more persistent
conditions governing the relation between the underlying assets may themselves
evolve. Evidence that recent relative-price movements contain predictive
information is provided by \textcite{BoonsPorrasPrado2019}. They show that
basis momentum, constructed from differences in past performance across
commodity futures maturities, predicts subsequent returns on futures-spreading
positions, and relate this predictability to supply-and-demand imbalances that
speculators and intermediaries have limited capacity to absorb. These
considerations motivate spread models in which a transient predictive
component interacts with a slower force governing longer-run convergence,
rather than treating all predictable price variation as the output of a single
mean-reverting state.
 
\medskip
{\emergencystretch=1.5em The longer-run convergence mechanism also raises an information problem:
the level toward which a spread is expected to converge need not be directly
observable. \textcite{LeePapanicolaou2016} study optimal pairs trading under
power utility with an unobservable mean-reversion level and quantify the
expected utility loss relative to full information. Their comparison
addresses the value of observing the hidden state. When several drift components 
operating at different time scales are themselves unobserved, the investor 
faces the additional difficulty of distinguishing their effects from 
observed price movements. The same widening of a spread
may reflect a temporary dislocation or trading pressure, a more persistent
change in the conditions governing convergence, or some combination of the
two, so the observed movement alone need not reveal which mechanism is
responsible. This makes partial information especially natural in a
multiscale spread model.\par}
 
\medskip
Once several forces are unobserved, the investor must infer the effects of these hidden forces from the
same observed price path. How that history is summarized becomes part of the
portfolio problem. A separate literature on technical trading rules approaches
this issue by specifying a summary of that history in advance. Moving averages
provide simple, implementable statistics of past prices without requiring
traders to observe the latent economic states governing predictable price
changes. \textcite{ZhuZhou2009} study their value as asset-allocation rules and 
their robustness under model uncertainty, while \textcite{LorigZhouZou2019} 
optimize within strategies constrained to depend on the deviation between 
log price and its exponential moving average (EMA). In these approaches, 
the moving-average form of the price-history statistic is specified in advance, 
although its parameters may be optimized as part of the portfolio problem.
 
\medskip
\textcite{ChenLee2023} reverse this direction of analysis. In a partially
observed model with a scalar latent drift factor, they show that the optimal
filtered estimate can be expressed in terms of current price and an EMA of
trailing price levels. The EMA-type signal is therefore selected endogenously
by the learning problem rather than imposed as a restriction on admissible
strategies. Its memory is determined by the persistence and uncertainty of the
hidden factor, giving the moving average a structural interpretation as the
statistic of past prices through which the investor learns about latent drift.
Their result provides a one-factor benchmark in which exponential price memory
emerges from optimal filtering itself.
 
\medskip
The multiscale partial-information literature already shows how jointly hidden 
predictive forces operating at different time scales can be filtered and 
incorporated into portfolio choice \parencite{FouquePapanicolaouSircar2017}. 
The endogenous-signal result of \textcite{ChenLee2023} shifts attention to a 
different aspect of that learning problem: the form in which the resulting 
estimate depends on observed price history. In their scalar model, filtering 
itself selects a recognizable EMA-type statistic of past prices. With several 
hidden time scales, however, the corresponding price-history statistic need 
not be obtained simply by attaching an exponential memory to each latent factor. 
Joint inference may alter not only the persistence of the resulting signal but 
its functional form. This motivates characterizing the price memory generated 
by a multiscale filter, determining when it reduces to familiar moving-average 
signals, and assessing whether the additional structure carries economically 
meaningful portfolio value.
 
\medskip
{\emergencystretch=1.5em The present paper addresses these questions in a continuous-time portfolio model 
for a traded spread with an arithmetic-valued price process whose drift is driven 
by two interacting latent factors, and for which the investor observes only prices. The fast factor enters the price drift directly
and captures shorter-lived variation in expected price changes. It mean
reverts toward a slowly moving level determined by the second factor, which
represents more persistent variation in the forces governing the fast state. Because neither factor is observed, the investor must
estimate them jointly from the price path, and the statistic used to summarize
that history is determined endogenously by the filtering problem. We show that
the long-run functional form of this statistic is governed by the balance
between effective fast--slow time-scale separation and the feedback created by
learning.\par}
 
\medskip
We first derive the associated Kalman--Bucy filter and solve the portfolio
optimization problem under logarithmic, power, and exponential utility over
arbitrary finite horizons, establishing optimality over price-adapted
admissible classes. In each case, the filtered fast factor, together with
current price, determines the conditional expected price change and therefore
enters directly into the investor's myopic position. Under logarithmic utility,
this component constitutes the entire optimal spread position. Under power and
exponential utility, an additional intertemporal hedging component can depend
on the full filtered state, including the estimate of the slow factor.
Consequently, characterizing how observed price history enters myopic demand
reduces to expressing the filtered fast factor as a functional of the price
path.
 
\medskip
We solve this representation problem exactly. At finite times, the filtered fast factor is a causal affine
functional of current price and the difference between fast and slow
filter adaptive EMA signals. Each signal exponentially weights price
history at the decay rate of its corresponding latent factor, while its
time-varying loading is determined by the Kalman gain and therefore reflects
how strongly the investor revises the filtered estimate in response to new
price information. We call the signed difference between these signals a
filter adaptive MACD signal. It describes not only how past prices enter optimal myopic demand,
but also how the investor's effective memory of those prices changes as
filtering uncertainty evolves. The estimate depends on this signal both directly
and through a feedback correction, which removes from it the price changes that
the investor had already forecast.
 
\medskip
We then determine the long-run form of this endogenous price memory. We prove
that the filtering-error covariance and Kalman gain converge to unique steady
states and obtain a complete classification of the resulting steady-state
signal. When effective fast--slow separation dominates the feedback induced by
filtering, the steady-state myopic portfolio depends on price history through current price and a weighted difference
of fast and slow EMA signals. We call this difference MACD-type, and we
characterize conventional MACD as an exact special case. When filtering feedback 
dominates separation, the corresponding price-history signal instead has damped 
oscillatory weights, so observations at successively older lags receive alternating 
positive and negative weights beneath a decaying envelope. Joint learning can therefore 
change the functional form of optimal price memory, not merely its persistence.
 
\medskip
The structural distinction between richer and simpler memory does not by itself
determine its economic importance. We therefore examine numerically, at illustrative
parameter values for which both time scales contribute to expected price
changes, what each element of the signal is worth to a logarithmic investor. In
these examples, most of the optimal growth rate is owed to price memory, and
the current price and a single EMA recover most of it. The second
EMA matters when price surprises inform the investor mainly about the slow
factor, and a conventional MACD then recovers most of what a single EMA loses.
The steady-state portfolio is a costly substitute for the finite-time one only
while the investor is still learning, and prices beyond the most recent year or
two add little. Finally, under power and exponential utility, we find that an investor who holds
only the myopic portfolio gives up a small part of the certainty-equivalent
gain at horizons up to ten years.
 
\subsection{Related literature}
 
The portfolio effects of predictable returns have been studied extensively in models with fully observable state variables. \textcite{KimOmberg1996} solve the portfolio problem with an observed Ornstein--Uhlenbeck risk premium and analyze how optimal positions depart from myopic demand as the investment horizon changes. \textcite{Liu2007} obtains power-utility portfolio rules for a broader class of observed Markov investment opportunities with quadratic structure. \textcite{HeLiLi2018} combine an observed mean-reverting fundamental component with momentum measured by a moving average of past returns over a fixed lookback horizon. They characterize optimal power-utility allocations, obtain an explicit logarithmic rule, and examine empirically how the lookback horizon affects performance.
 
\medskip
Our control problem belongs to the continuous-time literature on investment under partial information. In a Gaussian incomplete-information economy, \textcite{Detemple1986} shows how conditional means reduce the investor's problem to one with an observed state. \textcite{Lakner1995} treats utility maximization with investment and consumption adapted to observed prices, while \textcite{Lakner1998} derives explicit observation-history representations of trading strategies for Gaussian models, including multidimensional mean-reverting drift. \textcite{KaratzasZhao2001} study terminal-wealth optimization with Bayesian learning about a constant unknown drift with a given prior distribution. \textcite{Brendle2006} obtains an analytical power-utility solution under incomplete information and quantifies the utility cost relative to full observation. \textcite{BjorkDavisLanden2010} provide a general martingale treatment, deriving representations of optimal terminal wealth and strategies for logarithmic, power, and exponential utility. Our control results specialize this framework to an arithmetic spread with two coupled latent drift factors.
 
\medskip
A closely related application is the pairs-trading model of \textcite{LeePapanicolaou2016}, in which investors observe and trade both constituent assets while estimating a single hidden mean-reversion level. With the filter initialized at its steady error variance, they derive and verify power-utility portfolios, discuss the logarithmic limit, and establish bounded Riccati solutions for relative risk aversion greater than one. In our setting, both interacting drift factors must instead be inferred jointly from the price of one traded spread.
 
\medskip
Other contributions examine how the investor's information set shapes the estimation of hidden drift and the resulting portfolio decisions. \textcite{RiederBauerle2005} study logarithmic and power-utility portfolio choice when prices are the only observations of a hidden finite-state Markov drift, using posterior regime probabilities in the control state. \textcite{FreyGabihWunderlich2012} augment price observations in a hidden Markov model with noisy expert opinions arriving at random times. For multivariate Gaussian Ornstein--Uhlenbeck drift, \textcite{SassWestphalWunderlich2017} compare return-only, expert-only, combined, and full information, relating estimation covariance to logarithmic utility and the value of expert information. Unlike the finite-state regime models of Rieder and Bäuerle and Frey et al., our spread dynamics are driven by two continuously evolving Gaussian factors. Our observation structure also differs from that of Sass et al. Whereas the multivariate Gaussian model of Sass et al. provides a separate return observation for each hidden drift component, our investor must infer both factors through a single traded spread, while allowing correlation between price and factor shocks.
 
\medskip
A related line of work examines how learning about hidden drift connects optimal portfolio choice to technical trading rules. \textcite{Guasoni2006} connects partial information and technical analysis in a stock-price model with permanent Brownian shocks and a single mean-reverting fad. He derives the uninformed investor's drift as an explicit functional of past observations, solves the logarithmic investment problem, and quantifies the value of observing the fad. \textcite{BelHadjAyedLoeperAbergel2019} study a scalar hidden Ornstein--Uhlenbeck trend, express its steady-state Kalman estimate as a scaled exponential average of returns, and compare asymptotic expected logarithmic growth for the filter-based strategy and crossover rules based on finite-window geometric averages of prices, including under parameter misspecification. Whereas Guasoni and Bel Hadj Ayed et al. filter a single hidden drift state, our investor must jointly estimate a fast factor and the slow factor determining its moving mean. This additional hidden state permits qualitatively different forms of steady-state price memory, including damped oscillatory weights that cannot arise from a scalar constant-coefficient filtering equation.
 
\medskip
The closest multiscale antecedents are \textcite{FouquePapanicolaouSircar2015}, who study a rapidly mean-reverting hidden Gaussian drift, and \textcite{FouquePapanicolaouSircar2017}, who consider jointly unobserved fast and slow Gaussian factors with correlated shocks. They derive the Gaussian filters and use time-scale separation to develop asymptotic portfolio approximations for general concave utility. In a fully observed setting, \textcite{FouqueHu2020} establish multiscale performance expansions and asymptotic optimality within a specified strategy class when fast and slow factors drive both expected returns and volatility. Our signal analysis instead characterizes optimal price memory exactly, without imposing asymptotic separation.
 
\medskip
The time-series literature provides a complementary perspective on how latent-factor dynamics determine prediction from past observations. \textcite{Meddahi2002} derives analytical moving-average roots for the autoregressive moving-average representation of a noisy sum of two autoregressive factors. He relates the resulting best linear prediction recursion to steady-state Kalman filtering and examines complex moving-average roots and cancellations that reduce the representation's order. His financial application characterizes squared-return dynamics and volatility forecasts for discretely sampled continuous-time stochastic-volatility models. Our model instead concerns a continuously observed traded spread with constant price volatility and two hidden drift factors, so the learning problem determines conditional expected trading gains and optimal portfolio exposure rather than variance forecasts from discretely sampled squared returns.
 
\medskip
The paper is organized as follows. Section 2 introduces the two-factor latent
drift model and formulates the portfolio problem under partial information.
Section 3 derives the Kalman--Bucy filter and the optimal portfolios under 
logarithmic, exponential, and power utility, establishes solvability of 
the Riccati system associated with the power-utility HJB equation on every 
finite horizon, and states the verification theorem. Section 4 represents the
myopic portfolio in terms of a filter adaptive MACD signal. Section 5
classifies its steady-state dependence on effective fast--slow separation and 
filtering feedback. Section 6 studies the economic significance of the
signal representations, the value of past price observations, and
intertemporal hedging. Section 7 concludes. Appendices A--C contain,
respectively, the HJB and Riccati calculations, the proof of the verification
theorem, and the computational details for the numerical analysis.
 
%
%
 
\section{Model Setup}
 
Let $(\Omega,\F,\mathbb{F},\prob)$ be a filtered probability space satisfying the usual conditions and supporting a two-dimensional standard Brownian motion ${\bf W}=(W_t^F,W_t^S)^{\intercal}_{t\geq0}$ and a one-dimensional Brownian motion $W^P$. 
We assume that $W^P$ has instantaneous correlation vector $\rho:=(\rho_f,\rho_s)^{\intercal}$ with ${\bf W}$, so that
\begin{equation*}
    \dd\langle W^P,{\bf W}\rangle_t=\rho\dd t,
\end{equation*}
where the correlation coefficients satisfy $\|\rho\|_2<1$.
 
\subsection{Price Dynamics}
 
We consider a trader investing in a market with dynamics driven by the following constant coefficient linear-Gaussian diffusion model
\begin{equation}\label{model}
    \begin{cases}
        \dd P_t=\left(\lambda_p F_t-\kappa_p P_t\right)\dd t+\sigma_p\dd W_t^P\\[2pt]
        \dd\drift_t=\left(\mu-\kappa\drift_t\right)\dd t+\sigma\dd{\bf W}_t,\\[2pt]
        \drift_0\sim\mathcal{N}\left((f_0,s_0)^{\intercal},\hat{\Sigma}_0\right),\qquad P_0=p_0\in\R,
    \end{cases}
\end{equation}
 
where $P$ is an arithmetic-valued price process for a traded spread whose drift depends on
two latent diffusion factors operating at distinct time scales, collected in
the unobserved factor vector
$\drift=(F_t,S_t)^{\intercal}_{t\ge0}$; the vector
$\mu=(\mu_f,\mu_s)^{\intercal}$ and the matrices
\begin{equation*}
    \sigma=
    \begin{pmatrix}
        \sigma_f & \sigma_{f,s}\\
        \sigma_{f,s} & \sigma_s
    \end{pmatrix},
    \qquad
    \kappa=
    \begin{pmatrix}
        \kappa_f & -\lambda_f\\
        0 & \kappa_s
    \end{pmatrix}
\end{equation*}
determine the drift and volatility structure of the latent factors.
We assume that $\sigma$ and the initial covariance matrix $\hat{\Sigma}_0$ are positive definite, and that $\drift_0$ is independent of the Brownian motions $(W^P,{\bf W})$.
The model parameters are assumed to satisfy
\begin{equation*}
\begin{cases}
    \kappa_f>\kappa_s,\\
    \kappa_s,\kappa_p\ge0,\\
    \lambda_f,\lambda_p>0.
\end{cases}
\end{equation*}
 
\begin{remark}
The special case $\lambda_f=0$ is not considered here, since in that case the fast factor is decoupled from the slow factor and the model reduces to the single-factor partial-information dynamics studied in \textcite{ChenLee2023}.
\end{remark}
 
\medskip
In our model, the fast latent factor $F$ mean reverts at rate $\kappa_f>0$ toward the moving level $(\mu_f+\lambda_fS_t)/\kappa_f$, so the slow factor $S$ determines the level around which the fast state fluctuates, thereby influencing expected price changes indirectly. The slow factor itself mean reverts toward $\mu_s/\kappa_s$ at rate $\kappa_s$ when $\kappa_s>0$; when $\kappa_s=0$, it evolves as an arithmetic Brownian motion with drift $\mu_s$.  Since $\lambda_p>0$, a positive shock to $F$ increases the contribution $\lambda_p F_t$ to the price drift $\lambda_pF_t-\kappa_pP_t$, while the term $-\kappa_pP_t$ supplies a direct restoring force in the price dynamics when $\kappa_p>0$; when $\kappa_p=0$, this term is absent and $P$ is nonstationary.
 
\medskip
Since $\kappa_f>\kappa_s$, the model admits a natural separation between fast and slow adjustment speeds, with $F$ operating on a shorter time scale than $S$.
 
\subsection{The Optimal Portfolio Problem}\label{optimal_portfolio_problem}
 
We now formulate the associated portfolio optimization problem under partial information. The investor observes only the price process $P$, and trading strategies are therefore required to be adapted to the filtration generated by prices.
 
\medskip
Fix a deterministic finite horizon $T>0$, set $\T:=[0,T]$ and define the price filtration
\begin{equation*}
    \mathbb{F}^P := \{\F^P_t\}_{t\in\T}, 
    \qquad 
    \F^P_t := \sigma(P_s : 0\leq s\leq t),
\end{equation*}
augmented in the usual way.
 
\medskip
A trading strategy is a real-valued, $\mathbb{F}^P$-progressively measurable
process $\varphi = (\varphi_t)_{t\in\T}$ representing the number of
units held in the traded spread position. When $\varphi$ is an admissible trading strategy, we write $\varphi \in \mathcal{A}(\mathbb{F}^P)$. The precise admissible class depends on utility specification and is given in Definition \ref{admissible_strategy}.
 
\medskip
Given $\varphi \in \mathcal{A}(\mathbb{F}^P)$, the associated wealth process 
$V^\varphi = (V_t^\varphi)_{t\in\T}$ is defined by the self-financing condition
\begin{equation}\label{self_financing}
    \dd V_t^\varphi = \varphi_t \dd P_t, 
    \qquad 
    V_0^\varphi = v_0,
\end{equation}
for some deterministic initial wealth $v_0>0$. Implicit in \eqref{self_financing} are the assumptions of zero interest rates and frictionless markets.
 
\medskip
The investor seeks to maximize expected utility of terminal wealth:
\begin{equation*}
    \sup_{\varphi \in \mathcal{A}(\mathbb{F}^P)} 
    \E\!\left[ U(V_T^\varphi) \right],
\end{equation*}
where $U$ is a utility function defined on its natural wealth domain. Throughout the paper we consider logarithmic, power, and exponential utilities given by
\begin{align*}
    U_{\log}(v) &:= \log v, \\
    U_{\mathrm{pow}}(v) &:= \frac{v^{1-q}}{1-q}, \qquad q>1, \\
    U_{\exp}(v) &:= -e^{-p v}, \qquad p>0,
\end{align*}
where $q$ and $p$ denote the respective risk-aversion parameters.
 
%
%
 
\section{Optimal Trading Strategies}\label{filtering_optimality_section}
 
In this section, we derive candidate optimal trading strategies for each utility specification and verify that they are optimal. We begin by solving the optimal filtering problem associated with the partially observed system \eqref{model}. 
Since the dynamics are linear-Gaussian, the conditional distribution of the latent factor vector $\drift_t$ given the price filtration $\F_t^P$ is Gaussian.
The optimal filter is therefore characterized by its conditional mean
\begin{equation*}
    \kdrift_t=(\hat F_t,\hat S_t)^{\intercal}:=\E[\drift_t\mid\F_t^P]
\end{equation*}
and conditional filtering-error covariance
\begin{equation*}
    \kvar_t:=\E[(\drift_t-\kdrift_t)(\drift_t-\kdrift_t)^{\intercal}\mid\F_t^P],
\end{equation*}
with initial condition $\kvar_0=\hat\Sigma_0$.
 
\subsection{The Kalman-Bucy Equations}
 
The conditional mean and covariance matrix evolve according to the Kalman–Bucy filtering equations 
\parencite{Xiong2008, KaratzasZhao2001}. 
Before deriving these equations, we first rewrite the system \eqref{model} in the standard linear filtering form.
 
\medskip
Define the normalized observation process
\begin{equation}\label{normalized_price}
    \nprice_t:=\frac{\lambda_p}{\sigma_p}\int_0^t F_s\dd s+ W_t^P=\frac{1}{\sigma_p}(P_t-p_0)+\frac{\kappa_p}{\sigma_p}\int_0^t P_s\dd s,
\end{equation}
and introduce the processes
\begin{equation}\label{normalized_brownian}
    B_t^F:=\frac{W_t^F-\rho_f W_t^P}{\eta_f},\qquad
    B_t^S:=\frac{W_t^S-\rho_s W_t^P-\eta_{f,s} B_t^F}{\eta_s},
\end{equation}
where
\begin{equation*}
    \eta_f:=\sqrt{1-\rho_f^2},\qquad
    \eta_{f,s}:=-\frac{\rho_f\rho_s}{\sqrt{1-\rho_f^2}},\qquad
    \eta_s:=\sqrt{1-\frac{\rho_s^2}{1-\rho_f^2}}.
\end{equation*}
 
We also define the matrix
\begin{equation*}
    \eta:=\begin{pmatrix}
        \eta_f & 0\\
        \eta_{f,s} & \eta_s
    \end{pmatrix}.
\end{equation*}
 
\begin{remark}
    A direct calculation shows that $\rho$ and $\eta$ satisfy the identity
    \begin{equation*}
        \rho\rho^{\intercal}+\eta\eta^{\intercal}=I_2,
    \end{equation*}
    which reflects the decomposition of the factor noise into a component driven
    by $W^P$ and an independent component.
\end{remark}

\begin{lemma}
Let ${\bf B}:=(B_t^F,B_t^S)^{\intercal}_{t\geq0}$. Then ${\bf B}$ is a two-dimensional standard Brownian motion independent of $W^P$, and the system \eqref{model} can be written in the standard linear filtering form
\begin{equation}\label{standard_filter_form}
    \begin{cases}
        \dd \nprice_t=\frac{\lambda_p}{\sigma_p}F_t\dd t+\dd W_t^P,\\[6pt]
        \dd \drift_t=(\mu-\kappa\drift_t)\dd t+\sigma\rho\,\dd W_t^P+\sigma\eta\,\dd {\bf B}_t,\\[6pt]
        \drift_0\sim\mathcal N((f_0,s_0)^{\intercal},\hat\Sigma_0),
        \qquad\nprice_0=0.
    \end{cases}
\end{equation}
\end{lemma}
 
\begin{proof}
{\emergencystretch=1.5em Since $\rho_f^2+\rho_s^2<1$, the coefficients $\eta_f$ and $\eta_s$ are nonzero, so $(W^P,B^F,B^S)$ is well defined. Moreover, $(W^P,B^F,B^S)$ is a continuous Gaussian martingale starting from zero. A direct computation shows that its quadratic covariation matrix is the identity. Hence, by Lévy's characterization theorem, $(W^P,B^F,B^S)$ is a three-dimensional standard Brownian motion, and the first claim follows.\par}
 
    \medskip
    By \eqref{normalized_price}, the augmented filtration generated by the normalized observation process $\nprice$ coincides with the price filtration $\mathbb F^P$. Substituting the relations \eqref{normalized_brownian} into \eqref{model} and using the definition of $\nprice_t$ yields \eqref{standard_filter_form}.
\end{proof}
 
\medskip
The innovation process $\nu_t$ associated with the transformed system \eqref{standard_filter_form} is defined by
\begin{equation}\label{innovation}
    \nu_t:=\nprice_t-\frac{\lambda_p}{\sigma_p}\int_0^t \hat{F}_s\dd s,
\end{equation}
where $\hat{F}_t:=\kdrift^{(1)}_t$, and $(\nu_t)_{t\ge0}$ is an $\F_t^P$-Brownian motion.
 
\begin{corollary}
    The conditional mean $\kdrift$ and covariance matrix $\kvar$ satisfy the Kalman-Bucy filtering equations associated with the system \eqref{standard_filter_form}:
    \begin{equation}\label{kalman_bucy_equations}
    \begin{cases}
        \dd \kdrift_t=(\mu-\kappa\kdrift_t)\dd t+\left(\sigma\rho+\frac{\lambda_p}{\sigma_p}\kvar_t
        e_1\right)\dd \nu_t\\[6pt]
        \frac{\dd}{\dd t}\kvar_t=-\kvar_t\kappa^{\intercal}-\kappa\kvar_t+a-\left(\sigma\rho+\frac{\lambda_p}{\sigma_p}\kvar_te_1\right)\left(\sigma\rho+\frac{\lambda_p}{\sigma_p}\kvar_te_1\right)^{\intercal}\\[6pt]
        \kdrift_0=(f_0,s_0)^{\intercal}\in\R^2,\qquad\kvar_0=\hat{\Sigma}_0,
    \end{cases}
    \end{equation}
    where
    \begin{equation*}
        e_1:=(1,0)^{\intercal},\qquad a:=(\sigma\rho)(\sigma\rho)^{\intercal}+(\sigma\eta)(\sigma\eta)^{\intercal}=\sigma\sigma^{\intercal}.
    \end{equation*}
\end{corollary}

\medskip
Define the Kalman gain
\begin{equation}\label{drift_volatility}
    K_t:=\sigma\rho+\frac{\lambda_p}{\sigma_p}\kvar_t e_1,
\end{equation}
which is the volatility of the filtered drift $\kdrift$ in \eqref{kalman_bucy_equations}.


\subsection{Admissible Trading Strategies}\label{admissible_strategies_subsection}

Rearranging the definition of the innovation process in \eqref{innovation} gives the dynamics of the price in the investor's filtration,
\begin{equation}\label{observed_price_dynamics}
    \dd P_t=\pi_t\dd t+\sigma_p\dd\nu_t,\qquad\pi_t:=\lambda_p\hat F_t-\kappa_pP_t.
\end{equation}
The process $\pi$ is the conditional expected price change, and $\theta_t:=\pi_t/\sigma_p$ is the corresponding market price of risk. Define
\begin{equation}\label{pricing_density}
    Z_t:=\exp\left(-\int_0^t\theta_s\dd\nu_s-\frac12\int_0^t\theta_s^2\dd s\right),\qquad t\in\T.
\end{equation}
Since $\theta$ is an affine function of the Gaussian state $(P,\kdrift)$, the process $Z$ is a martingale by Lemma \ref{affine_exponential_martingale_lemma} in Appendix \ref{admissibility_verification_appendix}. It therefore defines a probability measure $\mathbb Q$ on $\F_T^P$, equivalent to $\prob$, by $\dd\mathbb Q/\dd\prob:=Z_T$. By Girsanov's theorem, $\nu_t^{\mathbb Q}:=\nu_t+\int_0^t\theta_s\dd s$ is a Brownian motion under $\mathbb Q$, and
\begin{equation}\label{pricing_measure_dynamics}
    \dd P_t=\sigma_p\dd\nu_t^{\mathbb Q}.
\end{equation}
The price is thus a martingale under $\mathbb Q$, which we call the \emph{pricing measure}.

\medskip
For utilities defined on the positive half-line, admissible wealth processes are required to remain positive \parencite{KramkovSchachermayer1999}. For exponential utility, which is defined on the whole real line, a supermartingale condition under the pricing measure is a standard admissibility device \parencite{Schachermayer2003}.

\begin{definition}[Admissible trading strategy]\label{admissible_strategy}
For $U\in\{U_{\log},U_{\mathrm{pow}},U_{\exp}\}$, an $\mathbb F^P$-progressively measurable strategy $\varphi$ belongs to $\mathcal A_U(\mathbb F^P)$ if the stochastic integral in \eqref{self_financing} is well defined and the wealth process $V^\varphi$ satisfies
\begin{itemize}
    \item[(i)] $V_t^\varphi>0$ for every $t\in\T$ almost surely, when $U\in\{U_{\log},U_{\mathrm{pow}}\}$;
    \item[(ii)] $V^\varphi$ is a $\mathbb Q$-supermartingale on $\T$, when $U=U_{\exp}$.
\end{itemize}
\end{definition}

Power and exponential utility are nonpositive, so their expected terminal utility is understood in $[-\infty,0]$. For logarithmic utility, the proof of Theorem \ref{verification_theorem} shows that the positive part of $\log V_T^\varphi$ is integrable for every admissible strategy, so that its expectation is well defined in $[-\infty,\infty)$.
 
 
\subsection{The Hamilton–Jacobi–Bellman Equation}
 
We now return to the partial-information portfolio problem introduced in Section \ref{optimal_portfolio_problem}.\footnote{For the classical PDE approach to dynamic programming used below,
see \textcite[Chapter 3]{Pham2009}.} Together with the filtering equations \eqref{kalman_bucy_equations}, the price dynamics \eqref{observed_price_dynamics} yield an equivalent fully observed Markovian formulation of the original portfolio problem in terms of the state process $(V_t^\varphi,P_t,\kdrift_t)$, where $V_t^\varphi$ denotes the current wealth associated to the strategy $\varphi$, $P_t$ the current price, and $\kdrift_t=(\hat F_t,\hat S_t)^{\intercal}$ the filtered drift estimate.
 
\medskip
If $\varphi$ denotes the investor's portfolio process, then the controlled state dynamics are
\begin{equation}\label{controlled_state_dynamics}
    \begin{cases}
        \dd V_t^{\varphi}=\varphi_t(\lambda_p\hat F_t-\kappa_p P_t)\dd t+\sigma_p\varphi_t\dd\nu_t,\\[6pt]
        \dd P_t=(\lambda_p\hat F_t-\kappa_p P_t)\dd t+\sigma_p\dd\nu_t,\\[6pt]
        \dd \kdrift_t=(\mu-\kappa\kdrift_t)\dd t+K_t\dd \nu_t.
    \end{cases}
\end{equation}
Since the Kalman gain $K_t$ is deterministic, this is a time-inhomogeneous Markov control problem driven by the innovation Brownian motion $\nu$.
 
\medskip
{\emergencystretch=2em Let $\cwealth \in\mathbb R$ denote the current wealth, $\cprice \in\mathbb R$ denote the current price, and $\cdrift:=(\cfast,\cslow)^{\intercal}\in\R^2$ denote a generic filtered drift state. Set
\begin{equation*}
    \mathcal O_{\log}=\mathcal O_{\mathrm{pow}}
    :=(0,\infty)\times\R^3,
    \qquad
    \mathcal O_{\exp}:=\R^4.
\end{equation*}
For $U\in\{U_{\log},U_{\mathrm{pow}},U_{\exp}\}$, write
$\mathcal O_U$ for the corresponding state domain above.
For $(t,\boldsymbol{x})\in\T\times\mathcal O_U$, where
$\boldsymbol{x}:=(\cwealth,\cprice,\cdrift)$, let
$X^{t,\boldsymbol{x};\varphi}$ denote the controlled state process
\eqref{controlled_state_dynamics} on $[t,T]$, initialized at
$X_t^{t,\boldsymbol{x};\varphi}=\boldsymbol{x}$. Let
$\mathcal A_U(t,\boldsymbol{x})$ denote the corresponding admissible controls
on $[t,T]$, defined as in Definition \ref{admissible_strategy} with $[t,T]$ in place of $\T$ and the density \eqref{pricing_density} started at time $t$,
and write $\E_{t,\boldsymbol{x}}$ for expectation under this initialized
system. Define the value function
\begin{equation*}
    \val_U(t,\cwealth,\cprice,\cdrift):=
    \sup_{\varphi\in\mathcal A_U(t,\boldsymbol{x})}
    \E_{t,\boldsymbol{x}}\!\left[U(V_T^{t,\boldsymbol{x};\varphi})\right],
\end{equation*}
with terminal condition $\val_U(T,\cwealth,\cprice,\cdrift)=U(\cwealth)$.
The original portfolio problem then has value $\val_U(0,v_0,p_0,\kdrift_0)$.\par}
 
\medskip
Fix a utility specification $U$, an open subset
$\mathcal O\subset\mathcal O_U$, and a test function
$\val\in C^{1,2}(\T\times\mathcal{O})$. The infinitesimal generator
\begin{equation*}
    \mathfrak L\val=\mathfrak L\val(t,\cwealth,\cprice,\cdrift,\varphi):\T\times\mathcal{O}\times\R\to\R
\end{equation*} 
corresponding to the controlled state dynamics in \eqref{controlled_state_dynamics} is defined by
\begin{align}
    \mathfrak L\val&:=
    \varphi(\lambda_p \cfast-\kappa_p \cprice)\,\partial_\cwealth\val
    +(\lambda_p \cfast-\kappa_p \cprice)\,\partial_\cprice\val
    +(\mu-\kappa \cdrift)\Cdot\nabla\val \nonumber\\
    &\quad
    \label{generic_hjb_expanded}
    +\frac12\sigma_p^2\varphi^2\,\partial_{\cwealth\cwealth}\val
    +\frac12\sigma_p^2\,\partial_{\cprice\cprice}\val
    +\frac12\mathrm{Tr}\!\left(K_tK_t^{\intercal}D^2\val\right) \\
    &\quad
    +\sigma_p^2\varphi\,\partial_{\cwealth\cprice}\val
    +\sigma_p\varphi\,K_t\Cdot\nabla\partial_\cwealth\val
    +\sigma_p K_t\Cdot\nabla\partial_\cprice\val,\nonumber
\end{align}
where $\nabla:=\nabla_{\cdrift}$ is the gradient operator and $D^2:=D^2_{\cdrift}$ is the Hessian operator for $\cdrift$, respectively. The dynamic programming principle formally yields the Hamilton–Jacobi–Bellman (HJB) equation associated with the portfolio problem:
\begin{equation}\label{generic_hjb}
    \begin{cases}
        \partial_t\val_U+\sup_{\varphi\in\mathbb R}\mathfrak L\val_U=0,\\[4pt]
        \val_U(T,\cwealth,\cprice,\cdrift)=U(\cwealth).
    \end{cases}
\end{equation}
We seek classical solutions of this equation and their associated optimal controls. The utility-specific reductions in Section \ref{utility-specific_reductions} below derive candidate value functions and their associated trading strategies for each utility specification. Before carrying out these reductions, however, it is useful to isolate a structural feature shared by all resulting candidate trading strategies.

\subsection{Myopic Demand and Intertemporal Hedging}
 
Following the standard decomposition of dynamic portfolio demand into myopic and intertemporal hedging components (see, e.g., \cite{Merton1973,KimOmberg1996,Liu2007}), we first define these two components in the present setting and then identify them in the generic candidate optimal strategy obtained from the formal HJB equation \eqref{generic_hjb}.
 
\medskip
Fix a utility specification $U\in\{U_{\log},U_{\mathrm{pow}},U_{\exp}\}$, a time $t<T$, and an initial state
$\boldsymbol x:=(\cwealth,\cprice,\cdrift)\in\mathcal O_U$.
For a fixed position $a\in\R$, consider the wealth process
$V^a$ initialized at $V_t^a=\cwealth$ and generated by holding
this position over a short interval. Let $\tau$ be the first exit time of the
state from a compact neighborhood of $\boldsymbol x$ contained
in $\mathcal O_U$.
 
\medskip
It\^o's formula, applied up to $\tau$, gives, for each fixed $a$,
\begin{align*}
    &\mathbb E_{t,\boldsymbol x}
    \left[
        U\bigl(V^a_{(t+h)\wedge\tau}\bigr)
    \right]
    -U(\cwealth)=h\left[
        a(\lambda_p\cfast-\kappa_p\cprice)U'(\cwealth)
        +\frac12\sigma_p^2a^2U''(\cwealth)
    \right]+o(h),\qquad h\downarrow0.
\end{align*}
The coefficient of $h$ is the instantaneous conditional
expected rate of change of utility. Maximizing this
coefficient evaluates the current investment opportunity using
$U$ itself, without valuing subsequent opportunities to trade.
 
\begin{definition}[Myopic demand]\label{myopic_demand}
For a state
$(\cwealth,\cprice,\cdrift)\in\mathcal O_U$,
the myopic demand is
\begin{align*}
    \varphi_U^{\mathrm M}(\cwealth,\cprice,\cdrift)
    &:=\operatorname*{arg\,max}_{a\in\R}
    \left\{
        a(\lambda_p\cfast-\kappa_p\cprice)U'(\cwealth)
        +\frac12\sigma_p^2a^2U''(\cwealth)
    \right\}\\
    &=-\frac{(\lambda_p\cfast-\kappa_p\cprice)U'(\cwealth)}{\sigma_p^2U''(\cwealth)}.
\end{align*}
\end{definition}
 
The difference between the candidate optimal strategy and the myopic demand
\begin{equation*}
    \varphi_U^{\mathrm H}(t,\cwealth,\cprice,\cdrift):=\varphi_U^*(t,\cwealth,\cprice,\cdrift)-\varphi_U^{\mathrm M}(\cwealth,\cprice,\cdrift)
\end{equation*}
therefore accounts for the intertemporal adjustment arising when future investment opportunities are taken into account, and is called the \emph{intertemporal hedging demand}.
 
\medskip
By completing the square with respect to $\varphi$ in the infinitesimal generator \eqref{generic_hjb_expanded}, we obtain the generic candidate optimal strategy
\begin{equation}\label{candidate_strategy}
    \varphi^*_U(t,\cwealth,\cprice,\cdrift):=-\frac{(\lambda_p\cfast-\kappa_p\cprice)\partial_{\cwealth}\val_U+\sigma_p^2\partial_{\cwealth\cprice}\val_U+\sigma_pK_t\Cdot\nabla\partial_{\cwealth}\val_U}{\sigma_p^2\partial_{\cwealth\cwealth}\val_U},
\end{equation}
whenever $\partial_{\cwealth\cwealth}\val_U<0$. By Proposition \ref{hjb_reduction} below, in all three utility specifications, the candidate value function $\val_U$ satisfies the risk-tolerance identity
\begin{equation*}
    \frac{\partial_{\cwealth}\val_U}{\partial_{\cwealth\cwealth}\val_U}
    =\frac{U'(\cwealth)}{U''(\cwealth)}.
\end{equation*}
 
\begin{corollary}[Intertemporal Hedging Demand]\label{intertemporal_hedging_formula}
    For a state $(\cwealth,\cprice,\cdrift)\in\mathcal O_U$, the generic intertemporal hedging demand is given by
    \begin{equation*}
    \varphi_U^{\mathrm H}(t,\cwealth,\cprice,\cdrift)=-\frac{\sigma_p^2\partial_{\cwealth\cprice}\val_U+\sigma_pK_t\Cdot\nabla\partial_{\cwealth}\val_U}{\sigma_p^2\partial_{\cwealth\cwealth}\val_U}.
    \end{equation*}
\end{corollary}
 
\subsection{Utility-Specific Reductions}\label{utility-specific_reductions}
 
We now specialize the formal HJB equation \eqref{generic_hjb} to the three utility specifications considered in this paper: logarithmic, power, and exponential utility. Throughout this section, let
\begin{equation*}
    \cpricedrift:=\begin{pmatrix}
        \cprice\\\cdrift
    \end{pmatrix}\in\R^3
\end{equation*}
denote the price-drift state variable. For each utility specification, we seek a candidate value function determined by a quadratic state functional of the form
\begin{equation*}
    \quadraticstate(t,\cpricedrift)=\cpricedrift^{\intercal}\quadraticcoef_t\cpricedrift+\linearcoef_t\Cdot\cpricedrift+\quadraticshift(t),
\end{equation*}
where the coefficient functions
\begin{equation*}
    \quadraticcoef:\T\to\mathrm{Sym}(3,\R),\qquad\linearcoef:\T\to\R^3,\qquad\quadraticshift:\T\to\R
\end{equation*}
are assumed to be $C^1$ and depend on the utility under consideration. For notational convenience, we also define the coefficient matrix and vectors
\begin{equation*}
    \K:=\begin{pmatrix}
        -\kappa_p & \lambda_p & 0\\
        0 & -\kappa_f & \lambda_f\\
        0 & 0 & -\kappa_s
    \end{pmatrix}\in\R^{3\times3},\qquad
    \augdrift:=\begin{pmatrix}
        0\\\mu
    \end{pmatrix}\in\R^3,\qquad
    \const:=\begin{pmatrix}
        -\kappa_p\\\lambda_p\\0
    \end{pmatrix}\in\R^3,
\end{equation*}
and the drift volatility state vector
\begin{equation*}
    \volatilitystate_t:=\begin{pmatrix}
        \sigma_p\\ K_t
    \end{pmatrix}\in\R^3.
\end{equation*}

\medskip
The homothetic preferences implied by logarithmic utility suggest an additively separable ansatz, the homogeneity of power utility a multiplicatively separable one, and the translation invariance of exponential utility an ansatz with exponential-affine separation of the wealth variable:
\begin{align}
    \psi_{\log}(t,\cwealth,\cpricedrift)&:=\log\cwealth+\quadraticstate(t,\cpricedrift),\label{log_candidate_value}\\
    \psi_{\mathrm{pow}}(t,\cwealth,\cpricedrift)&:=\frac{\cwealth^{1-q}}{1-q}
    \exp\!\bigl(\quadraticstate(t,\cpricedrift)\bigr),\label{power_candidate_value}\\
    \psi_{\exp}(t,\cwealth,\cpricedrift)&:=-\exp\!\bigl(-p\cwealth+\quadraticstate(t,\cpricedrift)\bigr),\label{exponential_candidate_value}
\end{align}
for the candidate value functions. In each case, the coefficients $(\quadraticcoef,\linearcoef,\quadraticshift)$ of $\quadraticstate$ are required to solve a backward ODE system on $\T$. Under logarithmic utility, the system is
\begin{equation}\label{log_reduced_ode}
\begin{cases}
    \dot \quadraticcoef_t+\quadraticcoef_t\K+\K^{\intercal}\quadraticcoef_t
    +\frac{1}{2\sigma_p^2}\const\const^{\intercal}=0,\\[4pt]
    \dot \linearcoef_t
    +\K^{\intercal}\linearcoef_t
    +2\quadraticcoef_t \augdrift=0,\\[4pt]
    \dot \quadraticshift(t)
    +\augdrift\Cdot\linearcoef_t
    +\volatilitystate_t^{\intercal}\quadraticcoef_t\volatilitystate_t=0,\\[4pt]
    \quadraticcoef(T)=0,\quad \linearcoef(T)=0,\quad \quadraticshift(T)=0.
\end{cases}
\end{equation}
Under power utility, it is the Riccati system
\begin{equation}\label{power_reduced_ode}
\begin{cases}
    \dot{\quadraticcoef}_t
    +\quadraticcoef_t\K+\K^{\intercal}\quadraticcoef_t
    +\dfrac{2}{q}\quadraticcoef_t\volatilitystate_t\volatilitystate_t^{\intercal}\quadraticcoef_t
    +\dfrac{1-q}{q\sigma_p}
    \bigl(
        \const\volatilitystate_t^{\intercal}\quadraticcoef_t
        +\quadraticcoef_t\volatilitystate_t\const^{\intercal}
    \bigr)
    +\dfrac{1-q}{2q\sigma_p^2}\const\const^{\intercal}=0,\\[10pt]
    \dot{\linearcoef}_t
    +\K^{\intercal}\linearcoef_t
    +2\quadraticcoef_t\augdrift
    +\dfrac{2}{q}(\volatilitystate_t^{\intercal}\linearcoef_t)\,\quadraticcoef_t\volatilitystate_t
    +\dfrac{1-q}{q\sigma_p}(\volatilitystate_t^{\intercal}\linearcoef_t)\,\const=0,\\[10pt]
    \dot{\quadraticshift}(t)
    +\augdrift\Cdot\linearcoef_t
    +\volatilitystate_t^{\intercal}\quadraticcoef_t\volatilitystate_t
    +\dfrac{1}{2q}(\volatilitystate_t^{\intercal}\linearcoef_t)^2=0,\\[10pt]
    \quadraticcoef(T)=0,\quad \linearcoef(T)=0,\quad \quadraticshift(T)=0.
\end{cases}
\end{equation}
Under exponential utility, it is
\begin{equation}\label{exp_reduced_ode}
\begin{cases}
    \dot{\quadraticcoef}_t
    +\quadraticcoef_t\K+\K^{\intercal}\quadraticcoef_t
    -\dfrac{1}{\sigma_p}
    \bigl(
        \const\volatilitystate_t^{\intercal}\quadraticcoef_t
        +\quadraticcoef_t\volatilitystate_t\const^{\intercal}
    \bigr)
    -\dfrac{1}{2\sigma_p^2}\const\const^{\intercal}=0,\\[10pt]
    \dot{\linearcoef}_t
    +\K^{\intercal}\linearcoef_t
    +2\quadraticcoef_t\augdrift
    -\dfrac{1}{\sigma_p}(\volatilitystate_t^{\intercal}\linearcoef_t)\,\const=0,\\[10pt]
    \dot{\quadraticshift}(t)
    +\augdrift\Cdot\linearcoef_t
    +\volatilitystate_t^{\intercal}\quadraticcoef_t\volatilitystate_t=0,\\[10pt]
    \quadraticcoef(T)=0,\quad \linearcoef(T)=0,\quad \quadraticshift(T)=0.
\end{cases}
\end{equation}

\medskip
The systems \eqref{log_reduced_ode} and \eqref{exp_reduced_ode} are linear. In each, the equation for $\quadraticcoef$ is a linear terminal-value problem with continuous coefficients. Once $\quadraticcoef$ is known, the equation for $\linearcoef$ is linear with continuous coefficients, and the remaining scalar equation determines $\quadraticshift$ by integration. Hence, each of these systems has a unique $C^1$ solution on $\T$. The equation for $\quadraticcoef$ in \eqref{power_reduced_ode} is quadratic, so its solvability on all of $\T$ requires an argument.

\begin{proposition}[Global Solvability of the Power Coefficient System]\label{power_coefficient_solvability}
For every finite horizon $T>0$, system \eqref{power_reduced_ode} has a unique
solution $(\quadraticcoef,\linearcoef,\quadraticshift)\in C^1(\T;\mathrm{Sym}(3,\R))\times C^1(\T;\R^3)\times C^1(\T;\R)$.
\end{proposition}

\begin{proposition}[HJB Reduction]\label{hjb_reduction}
For each $U\in\{U_{\log},U_{\mathrm{pow}},U_{\exp}\}$, let $(\quadraticcoef,\linearcoef,\quadraticshift)$ be the unique $C^1$ solution of the corresponding system \eqref{log_reduced_ode}, \eqref{power_reduced_ode}, or \eqref{exp_reduced_ode}, and define the candidate value function $\psi_U$ via \eqref{log_candidate_value}, \eqref{power_candidate_value}, or \eqref{exponential_candidate_value}. Then $\psi_U$ solves the HJB equation \eqref{generic_hjb} and satisfies
\begin{equation*}
    \partial_{\cwealth\cwealth}\psi_U<0,
    \qquad
    \frac{\partial_{\cwealth}\psi_U}{\partial_{\cwealth\cwealth}\psi_U}
    =\frac{U'(\cwealth)}{U''(\cwealth)},
\end{equation*}
and the candidate optimal strategies \eqref{candidate_strategy} have expressions
\begin{align}
    \varphi_{\log}^*(t,\cwealth,\cpricedrift)&=\frac{\cwealth}{\sigma_p^2}(\lambda_p\cfast-\kappa_p\cprice),\label{log_candidate_strategy}\\
    \varphi_{\mathrm{pow}}^*(t,\cwealth,\cpricedrift)&=\frac{\cwealth}{q\sigma_p^2}\left(\lambda_p\cfast-\kappa_p\cprice+\sigma_p\volatilitystate_t^{\intercal}\bigl(2\quadraticcoef_t\cpricedrift+\linearcoef_t\bigr)\right),\label{power_candidate_strategy}\\
    \varphi_{\exp}^*(t,\cpricedrift)&=\frac{1}{p\sigma_p^2}\left(\lambda_p\cfast-\kappa_p\cprice+\sigma_p\volatilitystate_t^{\intercal}\bigl(2\quadraticcoef_t\cpricedrift+\linearcoef_t\bigr)\right).\label{exponential_candidate_strategy}
\end{align}
\end{proposition}

The proofs of Propositions \ref{power_coefficient_solvability} and \ref{hjb_reduction} are given in Appendix \ref{hjb_riccati_appendix}.


\subsection{Optimal Portfolios}\label{optimal_portfolios_subsection}

The candidate strategies of Section \ref{utility-specific_reductions} were obtained from the formal HJB equation \eqref{generic_hjb}. The following theorem shows that they are optimal.

\begin{theorem}[Verification]\label{verification_theorem}
Fix a utility specification $U\in\{U_{\log},U_{\mathrm{pow}},U_{\exp}\}$, and let $\psi_U$ and $\varphi_U^*$ be the candidate value function and the candidate strategies of Proposition \ref{hjb_reduction}. Then
\begin{equation*}
    \psi_U(t,\cwealth,\cprice,\cdrift)=\val_U(t,\cwealth,\cprice,\cdrift)
\end{equation*}
for all $(t,\cwealth,\cprice,\cdrift)\in\T\times\mathcal O_U$. Moreover, the strategy $\varphi_t^*:=\varphi_U^*(t,V_t^*,P_t,\kdrift_t)$, with wealth process $V^*$, belongs to $\mathcal A_U(\mathbb F^P)$ and is optimal:
\begin{equation*}
    \E\bigl[U(V_T^*)\bigr]=\sup_{\varphi\in\mathcal A_U(\mathbb F^P)}\E\bigl[U(V_T^\varphi)\bigr]=\psi_U(0,v_0,p_0,\kdrift_0).
\end{equation*}
\end{theorem}

The proof is given in Appendix \ref{admissibility_verification_appendix}.

\medskip
Applying Definition \ref{myopic_demand} and Corollary \ref{intertemporal_hedging_formula} to the optimal strategies \eqref{log_candidate_strategy}, \eqref{power_candidate_strategy}, and \eqref{exponential_candidate_strategy}, we obtain their myopic and intertemporal hedging components.

\begin{corollary}\label{myopic_hedging_corollary}
    For each utility specification $U\in\{U_{\log},U_{\mathrm{pow}},U_{\exp}\}$, the myopic demand is
    \begin{equation*}
        \varphi_U^{\mathrm M}(\cwealth,\cprice,\cdrift)=G_U(\cwealth)(\lambda_p\cfast-\kappa_p\cprice),
        \qquad
        G_{\log}(\cwealth):=\frac{\cwealth}{\sigma_p^2},\quad
        G_{\mathrm{pow}}(\cwealth):=\frac{\cwealth}{q\sigma_p^2},\quad
        G_{\exp}(\cwealth):=\frac{1}{p\sigma_p^2},
    \end{equation*}
    and the intertemporal hedging demand is
    \begin{equation*}
        \varphi_{\log}^{\mathrm H}=0,\qquad
        \varphi_U^{\mathrm H}(t,\cwealth,\cpricedrift)=\sigma_pG_U(\cwealth)\,\volatilitystate_t^{\intercal}\bigl(2\quadraticcoef_t\cpricedrift+\linearcoef_t\bigr)
        \quad\text{for } U\in\{U_{\mathrm{pow}},U_{\exp}\},
    \end{equation*}
    where $(\quadraticcoef,\linearcoef)$ are the coefficients of the corresponding system \eqref{power_reduced_ode} or \eqref{exp_reduced_ode}.
\end{corollary}

The factor $G_U(\cwealth)=-U'(\cwealth)/(\sigma_p^2U''(\cwealth))$ is the investor's risk tolerance per unit of price variance. It multiplies the conditional expected price change $\lambda_p\cfast-\kappa_p\cprice$, so myopic demand has the same form under all three utilities. Hedging demand, by contrast, is absent under logarithmic utility and otherwise depends on the horizon and on the whole filtered state.

%
%
 
\section{Finite-Time Representation of the Myopic Portfolio}\label{finite_time_macd_section}

Evaluated along the observed price, the filtered state, and the wealth process $V$ of the strategy the investor follows, the myopic demand of Corollary \ref{myopic_hedging_corollary} becomes the process
\begin{equation}\label{myopic_portfolio}
    \varphi_t^{\mathrm M}:=\varphi_U^{\mathrm M}(V_t,P_t,\kdrift_t)=G_U(V_t)\bigl(\lambda_p\hat F_t-\kappa_pP_t\bigr),\qquad t\in\T,
\end{equation}
which we call the \emph{myopic portfolio}. Under logarithmic utility it is the whole of the optimal position. Wealth and the current price are observed directly, and under exponential utility $G_U$ is constant, so that wealth does not enter at all. Price history therefore enters the myopic portfolio only through the filtered fast factor $\hat F_t$. In this section we express $\hat F_t$, and with it the myopic portfolio, as a functional of the price path at every $t\in\T$.

\subsection{Price Surprises and the Filtered Fast Factor}

The filter revises its estimates only in response to the innovation process $\nu$. By \eqref{normalized_price} and \eqref{innovation},
\begin{equation}\label{price_surprise}
    \dd\nu_t=\dd\nprice_t-h\hat F_t\dd t,
    \qquad
    \sigma_p\dd\nprice_t=\dd P_t+\kappa_pP_t\dd t,
    \qquad
    h:=\frac{\lambda_p}{\sigma_p}.
\end{equation}
The increment $\sigma_p\dd\nprice_t$ is the price change net of the reversion $-\kappa_pP_t\dd t$, which the investor reads off the current price, and $h\hat F_t\dd t$ is the part of $\dd\nprice_t$ that the investor forecasts. We therefore call $\dd\nu_t$ a \emph{price surprise}: it is the unforecast part of the adjusted price change, measured in units of price volatility. The constant $h$ measures how strongly the fast factor shows through the noise in prices.

\medskip
By \eqref{kalman_bucy_equations}, a surprise $\dd\nu_u$ revises the two estimates by $K_u\dd\nu_u$, where $K$ is the Kalman gain \eqref{drift_volatility}: $K^{(1)}$ is the gain on the fast factor and $K^{(2)}$ the gain on the slow factor. Between surprises the estimates evolve as the latent factors do. Solving the linear equation for $\kdrift$ in \eqref{kalman_bucy_equations} by variation of constants gives
\begin{equation*}
    \kdrift_t=e^{-\kappa t}\kdrift_0+\int_0^te^{-\kappa(t-u)}\mu\dd u+\int_0^te^{-\kappa(t-u)}K_u\dd\nu_u,
\end{equation*}
where
\begin{equation*}
    e^{-\kappa x}=\begin{pmatrix}
        e^{-\kappa_fx} & \Delta_{f,s}(x)\\
        0 & e^{-\kappa_sx}
    \end{pmatrix},\qquad
    \Delta_{f,s}(x):=\frac{\lambda_f\left(e^{-\kappa_sx}-e^{-\kappa_fx}\right)}{\kappa_f-\kappa_s}.
\end{equation*}
Let $\Phi(x):=(e^{-\kappa_fx},\Delta_{f,s}(x))^{\intercal}$ denote the first row of $e^{-\kappa x}$. The first component of the preceding identity is
\begin{equation}\label{surprise_representation}
    \hat F_t=H(t)+\int_0^t\Phi(t-u)\Cdot K_u\dd\nu_u,
    \qquad
    H(t):=\Phi(t)\Cdot\kdrift_0+\int_0^t\Phi(x)\Cdot\mu\dd x.
\end{equation}
Since $\nu$ is an $\mathbb F^P$-Brownian motion and $K$ is deterministic, the stochastic integral has mean zero, and so $H(t)=\E[F_t]$: the function $H$ is the forecast of the fast factor that uses the prior alone.

\medskip
Equation \eqref{surprise_representation} describes $\hat F_t$ as a weighted sum of past price surprises, in which the surprise at time $u$ carries the weight
\begin{equation}\label{surprise_weight}
    \Phi(t-u)\Cdot K_u=e^{-\kappa_f(t-u)}K_u^{(1)}+\Delta_{f,s}(t-u)K_u^{(2)}.
\end{equation}
The two terms are the two paths by which a surprise reaches the estimate of the fast factor. The revision of the fast estimate acts at once and then decays at the rate $\kappa_f$. The revision of the slow estimate has no immediate effect on $\hat F$, because the slow factor moves expected price changes only by pulling the fast factor toward a new level; its effect builds at the rate $\kappa_f$ and fades at the rate $\kappa_s$. Accordingly $\Delta_{f,s}$ vanishes at lag zero, is positive at every other lag, and is proportional to the difference between a slow and a fast exponential. We shall see that this is the weighting an MACD applies to past price changes. Learning about a hidden factor that affects prices only through another hidden factor is thus the source of MACD-type structure in the model.

\medskip
Collecting the two exponentials in \eqref{surprise_weight}, we write
\begin{equation}\label{alpha_beta}
    \Phi(t-u)\Cdot K_u=e^{-\kappa_f(t-u)}\alpha_u-e^{-\kappa_s(t-u)}\beta_u,
    \qquad
    \beta_u:=-\frac{\lambda_f}{\kappa_f-\kappa_s}K_u^{(2)},
    \quad
    \alpha_u:=K_u^{(1)}+\beta_u.
\end{equation}
Since $\kvar$ solves the Riccati equation in \eqref{kalman_bucy_equations}, the gain $K$, and hence $\alpha$ and $\beta$, are deterministic and continuously differentiable on $\T$.

\subsection{Filter Adaptive EMA and MACD Signals}

Price surprises are not themselves a record of prices, since by \eqref{price_surprise} each of them subtracts the investor's own forecast. To express $\hat F$ through the price path, we first apply the weights in \eqref{surprise_representation} to the adjusted price changes $\dd\nprice$ in place of the surprises $\dd\nu$.

\begin{definition}[Filter adaptive EMA and MACD signals]\label{filter_adaptive_signals}
    For a decay rate $\gamma\geq0$ and a loading $\delta\in C^1(\T;\R)$, the \emph{filter adaptive EMA signal} is
    \begin{equation*}
        \mathcal E_t(\gamma,\delta):=\int_0^te^{-\gamma(t-u)}\delta_u\dd\nprice_u,\qquad t\in\T.
    \end{equation*}
    The \emph{filter adaptive MACD signal} is the signed difference between the fast and slow filter adaptive EMA signals with the loadings \eqref{alpha_beta},
    \begin{equation*}
        \mathcal M_t:=\mathcal E_t(\kappa_f,\alpha)-\mathcal E_t(\kappa_s,\beta)=\int_0^t\Phi(t-u)\Cdot K_u\dd\nprice_u,\qquad t\in\T.
    \end{equation*}
\end{definition}

Although it is written as a stochastic integral, $\mathcal E_t(\gamma,\delta)$ is a functional of the price path alone. Since $\delta$ is deterministic and continuously differentiable, \eqref{price_surprise} and integration by parts give
\begin{equation}\label{ema_pathwise}
    \mathcal E_t(\gamma,\delta)=\frac{\delta_t}{\sigma_p}P_t-\frac{e^{-\gamma t}\delta_0}{\sigma_p}P_0+\frac{1}{\sigma_p}\int_0^te^{-\gamma(t-u)}\bigl((\kappa_p-\gamma)\delta_u-\dot\delta_u\bigr)P_u\dd u.
\end{equation}
The signal thus combines the current price with an exponentially weighted record of past prices. Its rate $\gamma$ fixes how quickly an observation is forgotten. Its loading fixes how much the observation counted when it arrived: in the integral defining $\mathcal E_t(\gamma,\delta)$, the price change at time $u$ enters with the loading $\delta_u$ of that date. In $\mathcal M$ the loadings are built from the Kalman gain, so price changes observed while the investor was revising the estimates more strongly leave a larger imprint. This is the sense in which the signals adapt to the filter.

\medskip
{\emergencystretch=1.5em When the loading is constant, the signal reduces to familiar objects. For $\gamma\geq0$, let
\begin{equation}\label{price_ema}
    M_t^\gamma:=e^{-\gamma t}P_0+\gamma\int_0^te^{-\gamma(t-u)}P_u\dd u,\qquad t\in\T,
\end{equation}
denote the EMA of price with decay rate $\gamma$, started at $P_0$. For constant $\delta$ and $\gamma>0$, \eqref{ema_pathwise} becomes
\begin{equation}\label{ema_constant_loading}
    \mathcal E_t(\gamma,\delta)=\frac{\delta}{\sigma_p}\left(P_t+\frac{\kappa_p-\gamma}{\gamma}M_t^\gamma-\frac{\kappa_p}{\gamma}e^{-\gamma t}P_0\right),
\end{equation}
a combination of the current price and a single EMA, which is the form of the signal in the one-factor model of \textcite{ChenLee2023}. If the price has no reversion of its own, $\kappa_p=0$, then $\mathcal E_t(\gamma,\delta)=(\delta/\sigma_p)(P_t-M_t^\gamma)$ for every $\gamma\geq0$. The signal is then a multiple of the deviation of price from its EMA, and the difference between a slow and a fast signal with a common loading is a multiple of the conventional MACD:
\begin{equation}\label{conventional_macd_reduction}
    \mathcal E_t(\gamma_s,\delta)-\mathcal E_t(\gamma_f,\delta)
    =\frac{\delta}{\sigma_p}\bigl(M_t^{\gamma_f}-M_t^{\gamma_s}\bigr)
    =\frac{\delta}{\sigma_p}\int_0^t\bigl(e^{-\gamma_s(t-u)}-e^{-\gamma_f(t-u)}\bigr)\dd P_u,
    \qquad0\leq\gamma_s<\gamma_f.
\end{equation}
The last expression shows that an MACD weights past price changes by the difference between a slow and a fast exponential, which is the shape of $\Delta_{f,s}$ in \eqref{surprise_weight}.\par}

\medskip
The same structure is present in $\mathcal M$ at every finite time. Because $\mathcal E_t(\gamma,\delta)$ is linear in its loading, \eqref{alpha_beta} separates the signal according to the two components of the gain:
\begin{equation}\label{two_channel_signal}
    \mathcal M_t=\mathcal E_t\bigl(\kappa_f,K^{(1)}\bigr)
    +\frac{\lambda_f}{\kappa_f-\kappa_s}\Bigl[\mathcal E_t\bigl(\kappa_s,K^{(2)}\bigr)-\mathcal E_t\bigl(\kappa_f,K^{(2)}\bigr)\Bigr].
\end{equation}
The first channel is a single filter adaptive EMA signal at the fast rate, loaded by the gain on the fast factor. It records what prices reveal about the fast factor directly, and it would be the whole signal if $\lambda_f$ were zero, as in the one-factor model. The second channel is the difference between a slow and a fast filter adaptive EMA signal with a common loading, the gain on the slow factor. By \eqref{conventional_macd_reduction} it is the filter adaptive counterpart of a conventional MACD, and it records what prices reveal about the slow factor. The two channels in \eqref{two_channel_signal} are the two paths in \eqref{surprise_weight}.

\subsection{The Myopic Portfolio}

The signal $\mathcal M$ weights observed price changes exactly as \eqref{surprise_representation} weights price surprises. The two differ by the price changes the investor had forecast. Substituting \eqref{price_surprise} into \eqref{surprise_representation} gives
\begin{equation}\label{feedback_equation}
    \hat F_t=H(t)+\mathcal M_t-h\int_0^t\Phi(t-u)\Cdot K_u\hat F_u\dd u.
\end{equation}
The last term applies the same weights to the forecasts $h\hat F_u\dd u$. Since these forecasts are made with the estimate itself, \eqref{feedback_equation} is an equation for $\hat F$ rather than a formula. It is a linear Volterra equation with a continuous deterministic kernel, and the following lemma solves it.

\begin{lemma}[Feedback operator]\label{feedback_operator_lemma}
    For every $x\in C(\T;\R)$ there exists a unique $y\in C(\T;\R)$ such that
    \begin{equation}\label{volterra_equation}
        y_t=x_t-h\int_0^t\Phi(t-u)\Cdot K_uy_u\dd u,\qquad t\in\T.
    \end{equation}
    The map $\mathfrak B:C(\T;\R)\to C^1(\T;\R)$ defined by $\mathfrak B(x):=y-x$ is linear and causal, in that $\mathfrak B_t(x)$ depends on $x$ only through its values on $[0,t]$. Moreover,
    \begin{equation*}
        \mathfrak B_t(x)=-he^{-\kappa_ft}Y_f(t)+he^{-\kappa_st}Y_s(t),
    \end{equation*}
    where $\mathbf Y=(Y_f,Y_s)^{\intercal}$ is the unique solution of the linear system
    \begin{equation}\label{feedback_ode}
        \dot{\mathbf Y}(t)=\mathbf C_tx_t-h\mathbf G_t\mathbf Y(t),\qquad\mathbf Y(0)=0,
    \end{equation}
    with
    \begin{equation*}
        \mathbf C_t:=
        \begin{pmatrix}
            e^{\kappa_ft}\alpha_t\\
            e^{\kappa_st}\beta_t
        \end{pmatrix},
        \qquad
        \mathbf G_t:=
        \begin{pmatrix}
            \alpha_t & -e^{(\kappa_f-\kappa_s)t}\alpha_t\\
            e^{-(\kappa_f-\kappa_s)t}\beta_t & -\beta_t
        \end{pmatrix}.
    \end{equation*}
\end{lemma}

\begin{proof}
    Fix $x\in C(\T;\R)$. The system \eqref{feedback_ode} is linear with continuous coefficients, so it has a unique solution $\mathbf Y\in C^1(\T;\R^2)$. Define $y_t:=x_t-he^{-\kappa_ft}Y_f(t)+he^{-\kappa_st}Y_s(t)$. Then $\mathbf C_ty_t=\mathbf C_tx_t-h\mathbf G_t\mathbf Y(t)$, so that $\dot{\mathbf Y}(t)=\mathbf C_ty_t$ and
    \begin{equation}\label{feedback_states}
        Y_f(t)=\int_0^te^{\kappa_fu}\alpha_uy_u\dd u,\qquad Y_s(t)=\int_0^te^{\kappa_su}\beta_uy_u\dd u.
    \end{equation}
    By \eqref{alpha_beta}, $he^{-\kappa_ft}Y_f(t)-he^{-\kappa_st}Y_s(t)=h\int_0^t\Phi(t-u)\Cdot K_uy_u\dd u$, and hence $y$ solves \eqref{volterra_equation}. Conversely, if $y\in C(\T;\R)$ solves \eqref{volterra_equation}, then the functions in \eqref{feedback_states} solve \eqref{feedback_ode}, whose solution is unique; so $y$ is unique. The solution of \eqref{feedback_ode} is linear in $x$ and at time $t$ depends only on the values of $x$ on $[0,t]$, and it is continuously differentiable. The same therefore holds for $\mathfrak B(x)$.
\end{proof}

\begin{theorem}[Finite-time myopic portfolio]\label{macd_representation}
    For each utility specification $U\in\{U_{\log},U_{\mathrm{pow}},U_{\exp}\}$, the myopic portfolio \eqref{myopic_portfolio} satisfies, almost surely for all $t\in\T$,
    \begin{equation}\label{finite_time_myopic_portfolio}
        \varphi_t^{\mathrm M}=G_U(V_t)\Bigl(\lambda_p\bigl[\widetilde H(t)+\mathcal M_t+\mathfrak B_t(\mathcal M)\bigr]-\kappa_pP_t\Bigr),
    \end{equation}
    where $\mathcal M=\mathcal E(\kappa_f,\alpha)-\mathcal E(\kappa_s,\beta)$ is the filter adaptive MACD signal of Definition \ref{filter_adaptive_signals}, whose loadings $\alpha$ and $\beta$ are the linear functions \eqref{alpha_beta} of the Kalman gain; $\mathfrak B$ is the feedback operator of Lemma \ref{feedback_operator_lemma}, formed with $h=\lambda_p/\sigma_p$; and $\widetilde H:=H+\mathfrak B(H)$ is a deterministic function. Equivalently,
    \begin{equation}\label{finite_time_filter_representation}
        \hat F_t=\widetilde H(t)+\mathcal M_t+\mathfrak B_t(\mathcal M).
    \end{equation}
\end{theorem}

\begin{proof}
    By Definition \ref{filter_adaptive_signals}, substituting \eqref{price_surprise} into \eqref{surprise_representation} gives \eqref{feedback_equation}. The paths of $\hat F$ are continuous, and so are those of $\mathcal M$ by \eqref{ema_pathwise}. Hence, path by path, $\hat F$ is the solution of \eqref{volterra_equation} with $x=H+\mathcal M$, and Lemma \ref{feedback_operator_lemma} yields
    \begin{equation*}
        \hat F=H+\mathcal M+\mathfrak B(H+\mathcal M)=\widetilde H+\mathcal M+\mathfrak B(\mathcal M),
    \end{equation*}
    where the second equality uses the linearity of $\mathfrak B$. This is \eqref{finite_time_filter_representation}, and inserting it into \eqref{myopic_portfolio} gives \eqref{finite_time_myopic_portfolio}.
\end{proof}

Theorem \ref{macd_representation} shows that, at every finite time and under all three utilities, the myopic portfolio is the investor's risk tolerance multiplied by a causal affine functional of the current price and the filter adaptive MACD signal. The price path enters only through $\mathcal M$, since $\widetilde H$ is deterministic.

\medskip
The term $\mathfrak B(\mathcal M)$ is not a second signal. By \eqref{feedback_equation},
\begin{equation*}
    \mathfrak B_t(H+\mathcal M)=-h\int_0^t\Phi(t-u)\Cdot K_u\hat F_u\dd u,
\end{equation*}
so the operator removes from $\mathcal M$ the price changes that the investor had forecast, weighted exactly as the observed ones. It uses no information beyond the path of $\mathcal M$, and by Lemma \ref{feedback_operator_lemma} it requires no memory beyond the two auxiliary variables $Y_f$ and $Y_s$. It is nonetheless an integral part of the estimate, because it changes the rate at which past prices are forgotten.

%
%

\section{Steady-State Representation of the Myopic Portfolio}\label{steady_state_macd_section}

Theorem \ref{macd_representation} holds at every finite time, with loadings that vary with the Kalman gain and with the feedback left in the form of an operator. In this section we show that the gain converges to a constant, and that in the limit the feedback is absorbed into the rates and loadings of the signal. The form of the limiting signal is decided by whether the separation between the fast and slow estimates dominates the feedback that learning creates between them.

\subsection{The Steady-State Filter}

The feedback in \eqref{feedback_equation} can be absorbed into the dynamics of the filter. Writing the price surprise in \eqref{kalman_bucy_equations} as in \eqref{price_surprise} moves the investor's forecast into the drift of the estimates.

\begin{proposition}[Recursive form]\label{recursive_form}
    The filtered state satisfies
    \begin{equation}\label{finite_time_observation_feedback}
        \dd\kdrift_t
        =
        \bigl(\mu-A_c(t)\kdrift_t\bigr)\dd t+K_t\dd\nprice_t,
        \qquad
        A_c(t):=\kappa+hK_te_1^{\intercal}
        =\begin{pmatrix}\kappa_f+hK_t^{(1)}&-\lambda_f\\hK_t^{(2)}&\kappa_s\end{pmatrix}.
    \end{equation}
    The process $\mathbf Z_t:=\kdrift_t-(K_t/\sigma_p)P_t$ has continuously differentiable paths, and
    \begin{equation}\label{finite_time_price_state}
        \begin{cases}
            \displaystyle
            \dot{\mathbf Z}_t
            =
            \mu-A_c(t)\mathbf Z_t
            +\frac{1}{\sigma_p}
            \left(
                \bigl(\kappa_pI_2-A_c(t)\bigr)K_t-\dot K_t
            \right)P_t,\\[8pt]
            \displaystyle
            \mathbf Z_0
            =
            \kdrift_0-\frac{K_0}{\sigma_p}P_0,\\[8pt]
            \displaystyle
            \hat F_t
            =
            \frac{K_t^{(1)}}{\sigma_p}P_t+e_1^{\intercal}\mathbf Z_t.
        \end{cases}
    \end{equation}
    Consequently, for each utility specification $U\in\{U_{\log},U_{\mathrm{pow}},U_{\exp}\}$,
    \begin{equation}\label{recursive_myopic_portfolio}
        \varphi_t^{\mathrm M}=G_U(V_t)\Bigl(\bigl(hK_t^{(1)}-\kappa_p\bigr)P_t+\lambda_pe_1^{\intercal}\mathbf Z_t\Bigr).
    \end{equation}
\end{proposition}

\begin{proof}
    Substituting \eqref{price_surprise} into the first equation of \eqref{kalman_bucy_equations} and using $\hat F_t=e_1^{\intercal}\kdrift_t$ gives \eqref{finite_time_observation_feedback}. By \eqref{price_surprise}, $K_t\dd\nprice_t=(K_t/\sigma_p)\dd P_t+(\kappa_p/\sigma_p)K_tP_t\dd t$, and since $K$ is deterministic and continuously differentiable, $\dd(K_tP_t)=\dot K_tP_t\dd t+K_t\dd P_t$. The terms in $\dd P_t$ therefore cancel in $\dd\mathbf Z_t$, and replacing $\kdrift_t$ by $\mathbf Z_t+(K_t/\sigma_p)P_t$ in the drift gives the differential equation in \eqref{finite_time_price_state}. The remaining two lines of \eqref{finite_time_price_state} restate the definition of $\mathbf Z$, and \eqref{recursive_myopic_portfolio} follows from \eqref{myopic_portfolio} because $\lambda_p/\sigma_p=h$.
\end{proof}

The matrix $A_c(t)$ adds $hK_t^{(1)}$ to the decay rate of the fast estimate and couples the fast estimate to the slow one through $hK_t^{(2)}$. Its long-run behavior is that of the gain. Setting the time derivative in the Riccati equation of \eqref{kalman_bucy_equations} to zero gives the algebraic Riccati equation
\begin{equation}\label{steady_filter_are}
    0=-\kappa\bar\Sigma-\bar\Sigma\kappa^{\intercal}+\sigma\sigma^{\intercal}
    -\bigl(\sigma\rho+h\bar\Sigma e_1\bigr)\bigl(\sigma\rho+h\bar\Sigma e_1\bigr)^{\intercal}.
\end{equation}
For a solution $\bar\Sigma$, we write
\begin{equation}\label{steady_closed_loop_matrix}
    \bar K:=\sigma\rho+h\bar\Sigma e_1=:(k_1,k_2)^{\intercal},
    \qquad
    A_c:=\kappa+h\bar Ke_1^{\intercal}
    =\begin{pmatrix}\kappa_f+hk_1&-\lambda_f\\hk_2&\kappa_s\end{pmatrix}.
\end{equation}

\begin{proposition}[Steady-state gain]\label{steady_covariance_rates}
    \begin{enumerate}
        \item[(i)] Equation \eqref{steady_filter_are} has a unique positive-definite solution $\bar\Sigma$, and the eigenvalues of $A_c$ have positive real parts.
        \item[(ii)] For every prior covariance $\hat\Sigma_0$, we have $\kvar_t\to\bar\Sigma$, and hence $K_t\to\bar K$ and $A_c(t)\to A_c$, exponentially fast as $t\to\infty$.
    \end{enumerate}
\end{proposition}

\begin{proof}
    With $A:=-\kappa-h\sigma\rho e_1^{\intercal}$ and $Q:=\sigma(I_2-\rho\rho^{\intercal})\sigma^{\intercal}$, equation \eqref{steady_filter_are} reads $0=A\bar\Sigma+\bar\Sigma A^{\intercal}+Q-h^2\bar\Sigma e_1e_1^{\intercal}\bar\Sigma$, and the covariance equation in \eqref{kalman_bucy_equations} is its differential version. Since $\|\rho\|_2<1$ and $\sigma$ is nonsingular, $Q$ is positive definite, so $(A,Q^{1/2})$ is controllable, and the observability matrix of $(A,he_1^{\intercal})$ has determinant $h^2\lambda_f>0$. Hence the algebraic equation has a unique positive-semidefinite solution, this solution is stabilizing, and the differential equation converges to it from every positive-semidefinite initial condition \parencite{Wonham1968,KwakernaakSivan1972}. Here stabilizing means that the eigenvalues of $A-h^2\bar\Sigma e_1e_1^{\intercal}=-A_c$ have negative real parts. If $x^{\intercal}\bar\Sigma x=0$, then $\bar\Sigma x=0$ and the algebraic equation gives $x^{\intercal}Qx=0$, so $\bar\Sigma$ is positive definite. The convergence is exponential because $\Delta_t:=\kvar_t-\bar\Sigma$ solves $\dot\Delta_t=-A_c\Delta_t-\Delta_tA_c^{\intercal}-h^2\Delta_te_1e_1^{\intercal}\Delta_t$, whose linearization at zero is exponentially stable.
\end{proof}

Proposition \ref{steady_covariance_rates} suggests comparing the filter \eqref{finite_time_observation_feedback} with the one obtained by holding the gain at its limit. We define the \emph{steady-state filter} $\kdrift^\infty=(\hat F^\infty,\hat S^\infty)^{\intercal}$ by
\begin{equation}\label{steady_filter}
    \dd\kdrift_t^\infty=\bigl(\mu-A_c\kdrift_t^\infty\bigr)\dd t+\bar K\dd\nprice_t,
    \qquad\kdrift_0^\infty=\kdrift_0,
\end{equation}
and the \emph{steady-state myopic portfolio} by
\begin{equation}\label{steady_myopic_rule}
    \varphi_t^{\mathrm M,\infty}:=G_U(V_t)\bigl(\lambda_p\hat F_t^\infty-\kappa_pP_t\bigr),\qquad t\in\T.
\end{equation}

\subsection{The Steady-State Myopic Portfolio}

The steady-state filter has constant coefficients, so the weights it places on past prices are determined by the matrix $A_c$. We write
\begin{equation*}
    \bar r:=\tfrac12\operatorname{tr}A_c=\tfrac12(\kappa_f+\kappa_s+hk_1),\qquad
    d_c:=\det A_c=(\kappa_f+hk_1)\kappa_s+h\lambda_fk_2,\qquad
    \chi:=k_1\kappa_s+\lambda_fk_2,
\end{equation*}
and
\begin{equation*}
    g:=\kappa_f+hk_1-\kappa_s,\qquad c:=4h\lambda_fk_2,
\end{equation*}
so that $4(\bar r^2-d_c)=g^2-c$. The quantity $g$ is the gap between the diagonal entries of $A_c$, the separation between the decay rates of the fast and slow estimates. The quantity $c$ measures the feedback between the two estimates through the off-diagonal entries of $A_c$: the slow estimate is revised through $k_2$ and returns to the fast one through $\lambda_f$.

\begin{theorem}[Steady-state myopic portfolio]\label{steady_closed_loop_filter}
    For each utility specification $U\in\{U_{\log},U_{\mathrm{pow}},U_{\exp}\}$, the steady-state myopic portfolio \eqref{steady_myopic_rule} satisfies, almost surely for all $t\in\T$,
    \begin{equation}\label{steady_myopic_portfolio}
        \varphi_t^{\mathrm M,\infty}=G_U(V_t)\biggl(\lambda_p\Bigl[\bar H(t)+\int_0^tw(t-u)\dd\nprice_u\Bigr]-\kappa_pP_t\biggr),
        \qquad
        w(x):=e_1^{\intercal}e^{-A_cx}\bar K,
    \end{equation}
    where $\bar H(t):=e_1^{\intercal}\bigl(e^{-A_ct}\kdrift_0+\int_0^te^{-A_cx}\mu\dd x\bigr)$ is deterministic and converges exponentially to $e_1^{\intercal}A_c^{-1}\mu$. The weight $w$ takes one of two forms.
    \begin{enumerate}
        \item[(i)] \emph{Separation dominates, $g^2>c$.} The eigenvalues $r_f:=\bar r+\frac12\sqrt{g^2-c}$ and $r_s:=\bar r-\frac12\sqrt{g^2-c}$ of $A_c$ satisfy $0<r_s<r_f$, and
        \begin{equation*}
            w(x)=\bar\alpha e^{-r_fx}-\bar\beta e^{-r_sx},
            \qquad
            \bar\alpha:=\frac{k_1r_f-\chi}{r_f-r_s},\quad
            \bar\beta:=-\frac{\chi-k_1r_s}{r_f-r_s}.
        \end{equation*}
        Consequently,
        \begin{equation}\label{steady_chen_lee_representation}
            \varphi_t^{\mathrm M,\infty}=G_U(V_t)\Bigl(\lambda_p\bigl[\bar H(t)+\mathcal M_t^\infty\bigr]-\kappa_pP_t\Bigr),
            \qquad
            \mathcal M_t^\infty:=\mathcal E_t(r_f,\bar\alpha)-\mathcal E_t(r_s,\bar\beta),
        \end{equation}
        where we call $\mathcal M^\infty$ the \emph{steady-state MACD-type signal}.
        \item[(ii)] \emph{Feedback dominates, $g^2\leq c$.} Here $\bar r>0$, and with $\omega:=\sqrt{c-g^2}/2$,
        \begin{equation}\label{steady_complex_representation}
            w(x)=e^{-\bar rx}\Bigl[k_1\cos(\omega x)+(\chi-k_1\bar r)\frac{\sin(\omega x)}{\omega}\Bigr],
        \end{equation}
        where $\sin(\omega x)/\omega:=x$ when $\omega=0$. If $g^2<c$, then $w$ changes sign at lags spaced $\pi/\omega$ apart.
    \end{enumerate}
    Moreover, for all $t\in\T$,
    \begin{equation}\label{steady_myopic_demand}
        \E\biggl[\biggl(\frac{\varphi_t^{\mathrm M}-\varphi_t^{\mathrm M,\infty}}{G_U(V_t)}\biggr)^2\biggr]
        =\lambda_p^2\int_0^t\Bigl(e_1^{\intercal}e^{-A_c(t-u)}(K_u-\bar K)\Bigr)^2\dd u
        \leq Ce^{-\varepsilon t},
    \end{equation}
    for constants $C,\varepsilon>0$ that do not depend on the horizon $T$.
\end{theorem}

\begin{proof}
    Solving the linear equation \eqref{steady_filter} by variation of constants gives
    \begin{equation*}
        \kdrift_t^\infty=e^{-A_ct}\kdrift_0+\int_0^te^{-A_cx}\mu\dd x+\int_0^te^{-A_c(t-u)}\bar K\dd\nprice_u,
    \end{equation*}
    and inserting the first component into \eqref{steady_myopic_rule} gives \eqref{steady_myopic_portfolio}. By Proposition \ref{steady_covariance_rates}(i), the eigenvalues of $A_c$ have positive real parts, so $\bar H(t)$ converges exponentially to $e_1^{\intercal}A_c^{-1}\mu$.

    \medskip
    By the Cayley--Hamilton theorem, $A_c^2-2\bar rA_c+d_cI_2=0$. Hence $w$ is the solution of
    \begin{equation*}
        w''+2\bar rw'+d_cw=0,\qquad w(0)=k_1,\qquad w'(0)=-e_1^{\intercal}A_c\bar K=\chi-2\bar rk_1.
    \end{equation*}
    The roots of $r^2-2\bar rr+d_c$ are the eigenvalues of $A_c$, and its discriminant is $4(\bar r^2-d_c)=g^2-c$. If $g^2>c$, the roots $r_s<r_f$ are real, and positive by Proposition \ref{steady_covariance_rates}(i); the function $\bar\alpha e^{-r_fx}-\bar\beta e^{-r_sx}$ solves the equation and has the required initial values, because $\bar\alpha-\bar\beta=k_1$ and $r_s\bar\beta-r_f\bar\alpha=\chi-(r_f+r_s)k_1$. Definition \ref{filter_adaptive_signals} then gives $\int_0^tw(t-u)\dd\nprice_u=\mathcal M_t^\infty$. If $g^2\leq c$, the roots are $\bar r\pm i\omega$, and the function in \eqref{steady_complex_representation} solves the equation with the required initial values. If $g^2<c$, then $k_2\neq0$, so $k_1$ and $\chi-k_1\bar r$ do not both vanish, and $w$ is a nonzero multiple of $e^{-\bar rx}\cos(\omega x-\theta)$ for some $\theta$.

    \medskip
    For \eqref{steady_myopic_demand}, let $\boldsymbol e:=\kdrift-\kdrift^\infty$. Subtracting \eqref{steady_filter} from \eqref{finite_time_observation_feedback} and using \eqref{price_surprise} gives $\dd\boldsymbol e_t=-A_c\boldsymbol e_t\dd t+(K_t-\bar K)\dd\nu_t$ with $\boldsymbol e_0=0$. Thus $\boldsymbol e_t=\int_0^te^{-A_c(t-u)}(K_u-\bar K)\dd\nu_u$, while $(\varphi_t^{\mathrm M}-\varphi_t^{\mathrm M,\infty})/G_U(V_t)=\lambda_pe_1^{\intercal}\boldsymbol e_t$ by \eqref{myopic_portfolio} and \eqref{steady_myopic_rule}. The It\^o isometry gives the equality, and the bound follows from the exponential decay of $e^{-A_cx}$ and of $K_u-\bar K$ in Proposition \ref{steady_covariance_rates}.
\end{proof}

By \eqref{steady_myopic_demand}, the myopic portfolio converges exponentially fast to the steady-state myopic portfolio from every prior, in mean square per unit of risk tolerance $G_U(V_t)$. The two forms in (i) and (ii) are therefore the only long-run forms of the myopic portfolio, and which of them arises is decided by the model parameters alone.

\subsection{Conventional MACD}

The steady-state MACD-type signal can be written with the EMAs of price \eqref{price_ema}. When $g^2>c$, applying \eqref{ema_constant_loading} to both terms of $\mathcal M^\infty$ gives
\begin{equation}\label{steady_macd_decomposition}
    \lambda_p\mathcal M_t^\infty-\kappa_pP_t
    =(hk_1-\kappa_p)P_t+c_fM_t^{r_f}+c_sM_t^{r_s}+J(t),
    \qquad
    c_f:=h\bar\alpha\frac{\kappa_p-r_f}{r_f},\quad
    c_s:=-h\bar\beta\frac{\kappa_p-r_s}{r_s},
\end{equation}
where $J(t):=-h\kappa_p\bigl(\bar\alpha e^{-r_ft}/r_f-\bar\beta e^{-r_st}/r_s\bigr)P_0$ is deterministic and vanishes exponentially. By \eqref{steady_chen_lee_representation}, the steady-state myopic portfolio is thus a function of wealth, the current price, and a weighted difference of the fast and slow EMAs of price. This difference is a conventional MACD when the two EMAs carry equal and opposite weights. The following two results characterize exactly when a conventional MACD arises endogenously in the steady-state myopic portfolio.

\begin{lemma}[Uniqueness of price weights]\label{price_weight_uniqueness}
    Let $a\in\R$, let $k:[0,\infty)\to\R$ be real analytic, and let $\ell:\T\to\R$ be deterministic. If
    \begin{equation*}
        aP_t+\int_0^tk(t-u)P_u\dd u=\ell(t)\qquad\text{almost surely, for every }t\in\T,
    \end{equation*}
    then $a=0$ and $k$ vanishes identically.
\end{lemma}

\begin{proof}
    Denote the left-hand side by $X_t$. The integral in $X_t$ has continuously differentiable paths, with derivative $k(0)P_t+\int_0^tk'(t-u)P_u\dd u$. By \eqref{model}, the increment $X_{t+\varepsilon}-X_t$ is therefore the sum of $a\sigma_p(W_{t+\varepsilon}^P-W_t^P)$ and a term whose $L^2$ norm is of order $\varepsilon$. Since $X$ is deterministic, $|a|\sigma_p\sqrt\varepsilon$ is then of order $\varepsilon$ as $\varepsilon\downarrow0$, which forces $a=0$. Consequently $X$ has continuously differentiable paths, and $\dot X_t=k(0)P_t+\int_0^tk'(t-u)P_u\dd u$ is again deterministic. Repeating the argument with $(k(0),k')$ in place of $(a,k)$ shows inductively that every derivative of $k$ vanishes at zero, and $k$ vanishes identically because it is analytic.
\end{proof}

\begin{corollary}[Conventional MACD]\label{steady_corners}
    Suppose that $g^2>c$. There exist constants $\zeta\in\R$ and $\xi\neq0$, rates $0<\gamma_s<\gamma_f$, and a deterministic function $\ell$ such that
    \begin{equation*}
        \lambda_p\mathcal M_t^\infty-\kappa_pP_t=\zeta P_t+\xi\bigl(M_t^{\gamma_f}-M_t^{\gamma_s}\bigr)+\ell(t)\qquad\text{almost surely, for every }t\in\T,
    \end{equation*}
    if and only if
    \begin{equation}\label{conventional_macd_conditions}
        \kappa_p\chi=k_1d_c,\qquad\chi+(\kappa_p-2\bar r)k_1\neq0.
    \end{equation}
    In that case $\zeta=hk_1-\kappa_p$, $(\gamma_s,\gamma_f)=(r_s,r_f)$, $\ell=J$, and
    \begin{equation}\label{conventional_macd_portfolio}
        \varphi_t^{\mathrm M,\infty}=G_U(V_t)\Bigl(\lambda_p\bar H(t)+J(t)+(hk_1-\kappa_p)P_t+\xi\bigl(M_t^{r_f}-M_t^{r_s}\bigr)\Bigr),
        \qquad
        \xi=\frac{h\bigl(\chi+(\kappa_p-2\bar r)k_1\bigr)}{r_f-r_s}.
    \end{equation}
\end{corollary}

\begin{proof}
    By \eqref{ema_pathwise} with constant loadings, and by \eqref{price_ema},
    \begin{align*}
        \lambda_p\mathcal M_t^\infty-\kappa_pP_t&=(hk_1-\kappa_p)P_t+\int_0^tk(t-u)P_u\dd u+J_1(t),\\
        \xi\bigl(M_t^{\gamma_f}-M_t^{\gamma_s}\bigr)&=\xi\int_0^t\bigl(\gamma_fe^{-\gamma_f(t-u)}-\gamma_se^{-\gamma_s(t-u)}\bigr)P_u\dd u+J_2(t),
    \end{align*}
    with $J_1$ and $J_2$ deterministic and $k(x):=h\bigl(\kappa_pw(x)+w'(x)\bigr)=c_fr_fe^{-r_fx}+c_sr_se^{-r_sx}$. If the stated representation holds, Lemma \ref{price_weight_uniqueness} gives $\zeta=hk_1-\kappa_p$ and $k(x)=\xi(\gamma_fe^{-\gamma_fx}-\gamma_se^{-\gamma_sx})$ for all $x\geq0$. Since $\xi\neq0$ and exponentials with distinct rates are linearly independent, this holds if and only if $(\gamma_s,\gamma_f)=(r_s,r_f)$ and $c_f=-c_s=\xi$. Now $w(0)=k_1$, $w'(0)=\chi-2\bar rk_1$, and $\int_0^\infty w(x)\dd x=e_1^{\intercal}A_c^{-1}\bar K=\chi/d_c$, so that
    \begin{equation*}
        c_fr_f+c_sr_s=k(0)=h\bigl(\chi+(\kappa_p-2\bar r)k_1\bigr),
        \qquad
        c_f+c_s=\int_0^\infty k(x)\dd x=h\Bigl(\frac{\kappa_p\chi}{d_c}-k_1\Bigr).
    \end{equation*}
    Hence $c_f=-c_s$ if and only if $\kappa_p\chi=k_1d_c$, and then $c_f=h\bigl(\chi+(\kappa_p-2\bar r)k_1\bigr)/(r_f-r_s)$, which is nonzero if and only if the second condition in \eqref{conventional_macd_conditions} holds. This proves the necessity of \eqref{conventional_macd_conditions} and identifies $\zeta$, the rates, and $\xi$. Conversely, under \eqref{conventional_macd_conditions} the same identities give $c_f=-c_s=\xi\neq0$ in \eqref{steady_macd_decomposition}, which is the stated representation with $\zeta=hk_1-\kappa_p$ and $\ell=J$, and \eqref{conventional_macd_portfolio} follows from \eqref{steady_chen_lee_representation}.
\end{proof}

%
%

\section{Numerical Results}\label{sensitivity_analysis_section}

The representation theorems describe the form in which price history enters the myopic portfolio. In this section we measure what that form is worth to the investor. We ask how soon the steady-state myopic portfolio becomes an adequate substitute for the finite-time one, how much an investor loses by using a simpler signal, how much price history is worth keeping, and at which horizons intertemporal hedging matters.

\subsection{Baseline Economy and Performance Measure}\label{numerical_setup_subsection}

Time is measured in years of $252$ trading days. The baseline economy has
\begin{equation*}
    (\kappa_f,\kappa_s,\lambda_f,\kappa_p,\lambda_p,\sigma_p)=(12,\,0.6,\,12,\,0.25,\,1,\,0.4),
    \qquad
    \sigma=\begin{pmatrix}0.7&0\\0&0.4\end{pmatrix},
    \qquad
    \rho=(-0.2,\,0.3)^{\intercal},
\end{equation*}
with $\mu=0$, $p_0=0$, and $\kdrift_0=0$. Since $\lambda_f=\kappa_f$, the fast factor reverts to the slow factor itself, with a half-life of $15$ trading days, and the slow factor reverts to zero with a half-life of $14$ months. The slow factor accounts for $86\%$ of the stationary variance of $F$. These values are illustrative rather than calibrated, and are chosen so that both time scales contribute to expected price changes.

\medskip
We also vary the separation $\kappa_f/\kappa_s$ and the share of the stationary variance of $F$ due to the slow factor. In doing so we hold $\kappa_s$, the ratio $\lambda_f/\kappa_f$, and the stationary variance of $F$ at their baseline values, so that the amount of predictability is fixed and only its composition changes. Figure \ref{regime_figure} shows which of the two cases of Theorem \ref{steady_closed_loop_filter} arises in each of these economies. Feedback dominates when separation is weak and the slow factor is important; at the baseline share of $86\%$, this happens for $\kappa_f/\kappa_s<5.6$. The dashed curve is the set of economies in which the conditions \eqref{conventional_macd_conditions} hold, so that the steady-state signal is a conventional MACD, and the baseline lies close to it.

\begin{figure}[!t]
    \centering
    \includegraphics[width=.55\textwidth]{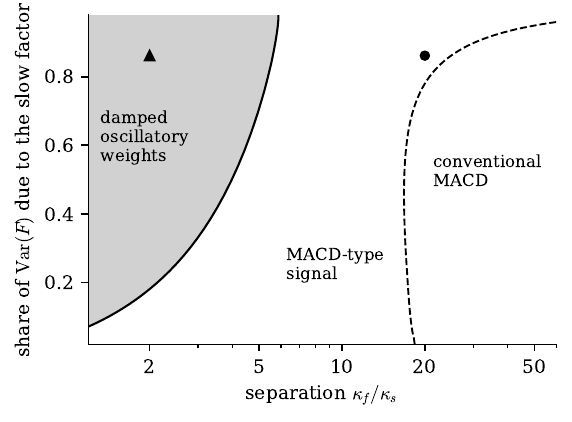}
    \caption{The long-run form of the signal in the myopic portfolio, as the separation $\kappa_f/\kappa_s$ and the share of the stationary variance of $F$ due to the slow factor vary. In the shaded region feedback dominates, $g^2<c$, and the weight on past price changes is the damped oscillation \eqref{steady_complex_representation}. Elsewhere separation dominates, $g^2>c$, and the signal is the steady-state MACD-type signal of \eqref{steady_chen_lee_representation}. On the dashed curve the conditions \eqref{conventional_macd_conditions} hold and the signal is a conventional MACD. The circle is the baseline and the triangle the weak-separation economy.}
    \label{regime_figure}
\end{figure}

\medskip
The \emph{weak-separation economy} has $\kappa_f/\kappa_s=2$ in place of $20$. At the baseline, $g^2=134$ and $c=33$, so separation dominates. In the weak-separation economy, $g^2=0.9$ and $c=3.5$, so feedback dominates.

\medskip
We measure performance by the expected logarithmic return of the logarithmic investor. Let $\widetilde F$ be an $\mathbb F^P$-progressively measurable process with $\E[\int_0^T\widetilde F_t^2\dd t]<\infty$, and let $\widetilde\varphi$ be the strategy that replaces $\hat F$ by $\widetilde F$ in \eqref{log_candidate_strategy}. With $\widetilde\pi_t:=\lambda_p\widetilde F_t-\kappa_pP_t$, this is $\widetilde\varphi_t=V_t^{\widetilde\varphi}\widetilde\pi_t/\sigma_p^2$. Its wealth is a stochastic exponential, so $\widetilde\varphi$ is admissible, and by \eqref{observed_price_dynamics} and It\^o's formula,
\begin{equation*}
    \log V_T^{\widetilde\varphi}-\log v_0
    =\frac{1}{\sigma_p^2}\int_0^T\Bigl(\widetilde\pi_t\pi_t-\frac12\widetilde\pi_t^2\Bigr)\dd t
    +\frac{1}{\sigma_p}\int_0^T\widetilde\pi_t\dd\nu_t.
\end{equation*}
The stochastic integral has mean zero, and $\widetilde\pi_t\pi_t-\frac12\widetilde\pi_t^2=\frac12\pi_t^2-\frac12(\pi_t-\widetilde\pi_t)^2$ with $\pi_t-\widetilde\pi_t=\lambda_p(\hat F_t-\widetilde F_t)$. Taking expectations therefore gives
\begin{equation}\label{log_signal_welfare_loss}
    \E\bigl[\log V_T^{\widetilde\varphi}\bigr]-\log v_0
    =\frac{1}{2\sigma_p^2}\E\biggl[\int_0^T\bigl(\lambda_p\hat F_t-\kappa_pP_t\bigr)^2\dd t\biggr]
    -\frac{\lambda_p^2}{2\sigma_p^2}\E\biggl[\int_0^T\bigl(\hat F_t-\widetilde F_t\bigr)^2\dd t\biggr].
\end{equation}
The first term is the expected logarithmic return of the optimal portfolio, which is the case $\widetilde F=\hat F$. The second is the expected logarithmic return that the investor gives up by trading on $\widetilde F$ instead of $\hat F$, and it is proportional to the mean-square error of $\widetilde F$ as an estimate of $\hat F$. Divided by $T$, we call the first term the \emph{growth rate} of the optimal portfolio and the second the \emph{loss} of the signal $\widetilde F$.

\medskip
When $\kappa_p$ and $\kappa_s$ are positive, the joint law of $(P_t,\drift_t,\kdrift_t)$ converges to a stationary law as $t\to\infty$. For the signals considered below, the two integrands in \eqref{log_signal_welfare_loss} then converge, and we call their limits the \emph{stationary growth rate} and the \emph{stationary loss}. At the baseline, the stationary growth rate of the optimal portfolio is $9.8\%$ per year, against $36.0\%$ for an investor who observes the factors. Appendix \ref{computational_appendix} describes how these and all later quantities are computed.

\subsection{Learning and the Steady-State Portfolio}

Suppose the investor begins with no information beyond the model, so that $\hat\Sigma_0$ is the stationary covariance of $\drift$. The Kalman gain then starts at $K_0=(0.23,0.44)^{\intercal}$ and falls to $\bar K=(0.07,0.28)^{\intercal}$: early price surprises carry more information than later ones. Figure \ref{learning_figure}(a) shows the loadings of the filter adaptive MACD signal. The loading $-\beta$ on the slow signal falls from $0.46$ to $0.29$, covering half of the distance in $69$ trading days, while $\alpha$ hardly moves.

\medskip
The steady-state myopic portfolio \eqref{steady_myopic_rule} applies the limiting gain from the start, and therefore gives too little weight to early surprises. Under logarithmic utility it is the strategy $\widetilde\varphi$ of Section \ref{numerical_setup_subsection} with $\widetilde F=\hat F^\infty$, so by \eqref{steady_myopic_demand} and \eqref{log_signal_welfare_loss}, its loss over $\T$ is
\begin{align*}
    \frac{\lambda_p^2}{2\sigma_p^2T}\int_0^T\E\Bigl[\bigl(\hat F_t-\hat F_t^\infty\bigr)^2\Bigr]\dd t
    &=\frac{1}{2\sigma_p^2 T}\int_0^T\E\left[\left(\frac{\varphi_t^{\mathrm M}-\varphi_t^{\mathrm M,\infty}}{G_{\log}(V_t)}\right)^2\right]\dd t
    \\&=\frac{\lambda_p^2}{2\sigma_p^2T}\int_0^T\!\!\int_0^t\Bigl(e_1^{\intercal}e^{-A_c(t-u)}(K_u-\bar K)\Bigr)^2\dd u\dd t.
\end{align*}
Figure \ref{learning_figure}(c) reports this loss as a share of the growth rate of the optimal portfolio. It is $14\%$ over one year, $8\%$ over two years, $3\%$ over five years, and $1\%$ over ten years. The steady-state portfolio is thus a poor substitute for the finite-time one while the investor is still learning, and an accurate one afterwards. An investor whose prior is already $\bar\Sigma$ loses nothing by holding the gain fixed.

\begin{figure}[!htb]
    \centering
    \includegraphics[width=\textwidth]{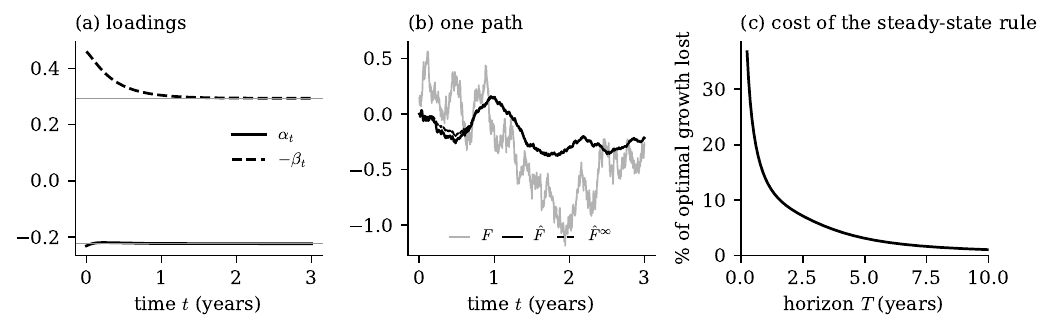}
    \caption{Learning at the baseline, from the prior $\hat\Sigma_0$ equal to the stationary covariance of $\drift$. Panel (a): the loadings $\alpha$ and $-\beta$ of \eqref{alpha_beta}, with their limits in grey. Panel (b): one simulated path of the fast factor, its filtered estimate, and its steady-state estimate \eqref{steady_filter}. Panel (c): the loss of the steady-state myopic portfolio over $\T$, as a percentage of the growth rate of the optimal portfolio over the same horizon.}
    \label{learning_figure}
\end{figure}

\subsection{The Value of the Second Time Scale}

We now compare the optimal portfolio, in the stationary state, with portfolios built on simpler signals. Each uses the current price together with a restricted record of past prices:
\begin{align*}
    \text{price only:}&\qquad\widetilde F_t=c_PP_t,\\
    \text{price and one EMA:}&\qquad\widetilde F_t=c_PP_t+c_MM_t^{\gamma},\\
    \text{price and conventional MACD:}&\qquad\widetilde F_t=c_PP_t+a\bigl(M_t^{\gamma_f}-M_t^{\gamma_s}\bigr),\\
    \text{price and two EMAs:}&\qquad\widetilde F_t=c_PP_t+c_fM_t^{\gamma_f}+c_sM_t^{\gamma_s}.
\end{align*}
The second is the form of the signal in the one-factor model of \textcite{ChenLee2023}, the third is the form in Corollary \ref{steady_corners}, and the fourth is the form in \eqref{steady_macd_decomposition}. Within each class we choose the coefficients and the rates to minimize the stationary loss, so the comparison favors the simpler signals.

\begin{table}[!t]
    \centering
    {\small
    \begin{tabular}{@{}lrrrlrrrl@{}}
        \hline
        & \multicolumn{4}{c}{Baseline, $\kappa_f/\kappa_s=20$} & \multicolumn{4}{c}{Weak separation, $\kappa_f/\kappa_s=2$}\\
        Signal & Growth & \multicolumn{2}{c}{Loss} & Rates & Growth & \multicolumn{2}{c}{Loss} & Rates\\
        \hline
        Price only & 1.73 & 804 & 82.3\% & & 1.43 & 753 & 84.0\% & \\
        Price and one EMA & 9.38 & 39 & 4.0\% & 1.04 & 8.52 & 44 & 4.9\% & 0.70\\
        Price and conventional MACD & 9.77 & 0.1 & 0.01\% & 12.4, 1.36 & 8.81 & 15 & 1.7\% & 3.08, 1.21\\
        Price and two EMAs & 9.77 & 0 & 0\% & 11.4, 1.38 & 8.93 & 2.8 & 0.3\% & $\to1.60$\\
        Optimal & 9.77 & & & & 8.96 & & & \\
        Full information & 35.97 & & & & 31.30 & & & \\
        \hline
    \end{tabular}
    }
    \caption{Stationary growth rate of the logarithmic investor under signals of increasing complexity. Growth is in percent per year, the loss relative to the optimal portfolio is in basis points per year and as a percentage of the optimal growth rate, and the fitted rates are per year.}
    \label{ladder_table}
\end{table}

\medskip
Table \ref{ladder_table} reports the results. Price memory is worth most of what the investor earns: with the current price alone, the growth rate falls from $9.8\%$ to $1.7\%$ at the baseline. A single EMA recovers all but $4\%$ of the optimal growth rate, or $39$ basis points per year, with a fitted rate of $1.04$, close to the slow rate of the optimal signal. The second EMA recovers the remainder. A conventional MACD is almost as good as two unrestricted EMAs, because the optimal signal puts nearly equal and opposite weights on its two EMAs: at the baseline, \eqref{steady_macd_decomposition} reads
\begin{equation*}
    \lambda_p\mathcal M_t^\infty-\kappa_pP_t=-0.08\,P_t+0.63\,M_t^{r_f}-0.67\,M_t^{r_s}+J(t),
    \qquad(r_f,r_s)=(11.4,\,1.38).
\end{equation*}
The two rates are those of daily EMAs with spans of $44$ and $366$ trading days. In the weak-separation economy, no combination of two EMAs reproduces the optimal signal, but two EMAs come within $0.3\%$ of it. The weight \eqref{steady_complex_representation} has $\bar r=1.07$ and $\omega=0.81$, so its successive extrema shrink by the factor $e^{-\pi\bar r/\omega}=0.016$, and so heavily damped an oscillation is hard to distinguish from two exponentials. The loss of $0.3\%$ is not attained by any pair of rates: it is approached as both rates tend to $1.60$, while the weights on the two EMAs grow large and of opposite sign.

\medskip
Figure \ref{value_figure}(a) shows the loss from a single EMA across the economies of Figure \ref{regime_figure}. It is below $1\%$ when the slow factor contributes less than a third of the variance of $F$ and separation is moderate, about $4\%$ at the baseline share whatever the separation, and reaches $10\%$ when separation is strongest. The boundary between the two regimes of Theorem \ref{steady_closed_loop_filter} leaves no trace in the loss.

\medskip
What determines the value of the second time scale is how the investor learns about the fast factor. Figure \ref{value_figure}(b) varies the correlation $\rho_f$ between price shocks and shocks to the fast factor. As $\rho_f$ falls, the gain $k_1$ falls with it, from $0.20$ at $\rho_f=0$ to $-0.21$ at $\rho_f=-0.6$, and the weight $w$ of Theorem \ref{steady_closed_loop_filter}, which equals $k_1$ at lag zero, becomes more hump-shaped. The loss from a single EMA rises to $7\%$ at $\rho_f=-0.3$ and $21\%$ at $\rho_f=-0.6$. A conventional MACD loses $0.2\%$ and $6\%$ at the same values, and is exact at $\rho_f=-0.22$. In the other direction, a single EMA is exact at $\rho_f=0.13$, where $\bar\alpha=0$ and the weight in Theorem \ref{steady_closed_loop_filter}(i) is a single exponential. When price surprises inform the investor about the fast factor directly, one EMA is therefore enough. When they inform the investor mainly about the slow factor, the optimal signal has the shape of an MACD, and a single EMA cannot reproduce it.

\begin{figure}[!htb]
    \centering
    \includegraphics[width=.92\textwidth]{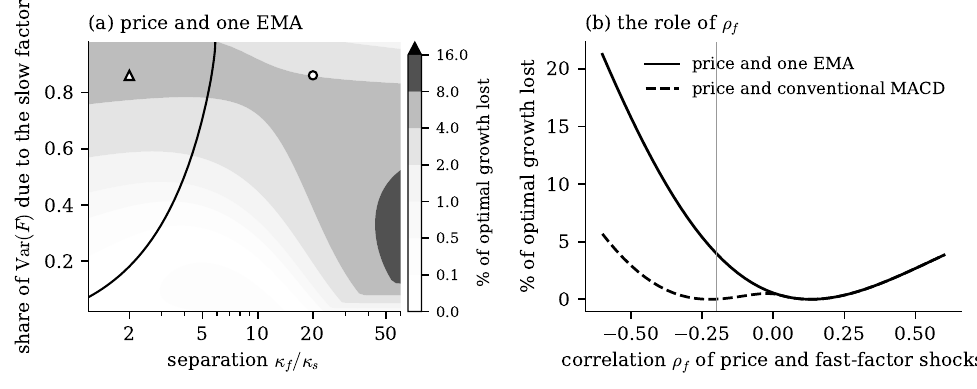}
    \caption{The stationary loss of simpler signals, as a percentage of the optimal growth rate. Panel (a): price and one EMA, over the economies of Figure \ref{regime_figure}; the solid curve is the boundary $g^2=c$. Panel (b): price and one EMA, and price and a conventional MACD, as $\rho_f$ varies from its baseline value, marked by the vertical line.}
    \label{value_figure}
\end{figure}

\FloatBarrier

\subsection{The Value of Past Prices}

The rates $r_f$ and $r_s$ describe how fast the weight on an old price decays, but not how much the investor would lose without it. For $L>0$, let
\begin{equation*}
    \hat F_t^{(L)}:=\E\bigl[F_t\mid P_u,\ t-L\leq u\leq t\bigr]
\end{equation*}
be the estimate of an investor who retains only the last $L$ years of prices. The error $F_t-\hat F_t$ is uncorrelated with every square-integrable $\F_t^P$-measurable random variable. Since $\hat F_t^{(L)}$ depends only on prices up to time $t$, this includes $\hat F_t-\hat F_t^{(L)}$, and therefore
\begin{equation*}
    \E\Bigl[\bigl(F_t-\hat F_t^{(L)}\bigr)^2\Bigr]
    =\E\Bigl[\bigl(F_t-\hat F_t\bigr)^2\Bigr]
    +\E\Bigl[\bigl(\hat F_t-\hat F_t^{(L)}\bigr)^2\Bigr].
\end{equation*}
Because the model is Gaussian, conditional variances are deterministic. The left-hand side is thus the conditional variance of $F_t$ given the last $L$ years of prices, and the first term on the right is the first diagonal entry of $\kvar_t$, which equals $\bar\Sigma_{11}$ in the stationary state. By \eqref{log_signal_welfare_loss}, the stationary loss of $\hat F^{(L)}$ is therefore
\begin{align}
    \frac{\lambda_p^2}{2\sigma_p^2}\E\Bigl[\bigl(\hat F_t-\hat F_t^{(L)}\bigr)^2\Bigr]
    &=\frac{\lambda_p^2}{2\sigma_p^2}\biggl(\E\Bigl[\bigl(F_t-\hat F_t^{(L)}\bigr)^2\Bigr]-\E\Bigl[\bigl(F_t-\hat F_t\bigr)^2\Bigr]\biggr)\nonumber\\
    &=\frac{\lambda_p^2}{2\sigma_p^2}\Bigl(\operatorname{Var}\bigl(F_t\mid P_u,\ t-L\leq u\leq t\bigr)-\bar\Sigma_{11}\Bigr).
    \label{trailing_projection_identity}
\end{align}

\medskip
Figure \ref{history_hedging_figure}(a) shows the share of the optimal growth rate that such an investor earns. At the baseline, the current price alone earns $18\%$ of it. Half is earned with $71$ trading days of history, $90\%$ with $223$ days, and $99\%$ with $436$ days. Stronger separation shortens the history needed for half of the value, but not for all of it: with $\kappa_f/\kappa_s=60$ the three figures are $58$, $213$, and $436$ days. Under weak separation they are $132$, $360$, and $554$ days. In all three economies, prices more than about two years old add less than $1\%$ to the investor's growth rate.

\subsection{Hedging Demand and the Horizon}

Under power and exponential utility, the optimal portfolio adds to the myopic portfolio the hedging demand of Corollary \ref{myopic_hedging_corollary}, which depends on the horizon and on the whole filtered state. We measure what an investor who holds only the myopic portfolio gives up. For a strategy $\varphi$, let the certainty equivalent $c_T(\varphi)$ be defined by $U(c_T(\varphi))=\E[U(V_T^\varphi)]$, and define the certainty-equivalent gain as $\log(c_T(\varphi)/v_0)$ under power utility and $c_T(\varphi)-v_0$ under exponential utility. The investor starts from $P_0=0$ and $\kdrift_0=0$ with $\hat\Sigma_0=\bar\Sigma$, so that learning plays no role, and the expected utility of the myopic portfolio is obtained from a coefficient system of the same type as \eqref{power_reduced_ode} and \eqref{exp_reduced_ode}, given in Appendix \ref{computational_appendix}.

\medskip
Figure \ref{history_hedging_figure}(b) shows the share of the optimal certainty-equivalent gain that the myopic portfolio gives up, for exponential utility and for power utility $U_{\mathrm{pow}}(v)=v^{1-q}/(1-q)$ with relative risk aversion $q=2$, $5$, and $10$. Over one year it is below $1\%$ for all four investors. Over five years it is $0.8\%$ for power utility with $q=2$, $1.7\%$ with $q=5$, $2.0\%$ with $q=10$, and $2.4\%$ for exponential utility. Over ten years the four figures are $2.1\%$, $4.6\%$, $5.5\%$, and $6.5\%$. The cost of ignoring hedging demand grows with the horizon and with risk aversion, but at horizons of a few years it is small: the myopic portfolio, and with it the MACD-type signal, accounts for nearly all of what the investor gains from trading.

\begin{figure}[!htb]
    \centering
    \includegraphics[width=.92\textwidth]{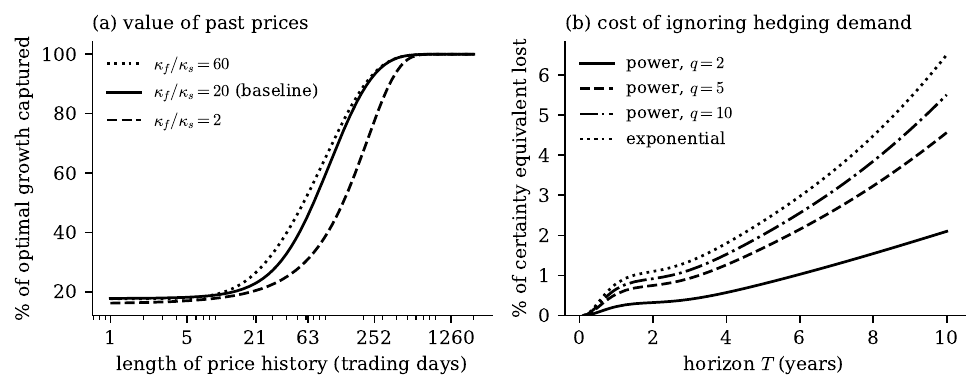}
    \caption{Panel (a): the stationary growth rate of a logarithmic investor who conditions only on the last $L$ years of prices, as a percentage of the optimal growth rate, at the baseline and at two other values of $\kappa_f/\kappa_s$. Panel (b): the share of the optimal certainty-equivalent gain that the myopic portfolio gives up over the horizon $T$, at the baseline, for power utility with relative risk aversion $q=2,5,10$ and for exponential utility.}
    \label{history_hedging_figure}
\end{figure}

%
%
 
\section{Conclusion}

We have shown that MACD signals can arise endogenously from a partially observable portfolio optimization problem. When expected price changes are driven by a fast latent factor that reverts toward a slow one, and the investor observes only prices, the myopic part of the optimal portfolio depends on price history through a filter adaptive MACD signal. In the long run, the signal is MACD-type when the separation between the decay rates of the two estimates dominates the feedback that learning creates between them, and has damped oscillatory weights when feedback dominates. In the first case, we have characterized exactly when the myopic portfolio is a function of the current price and a conventional MACD, providing a mathematical foundation for one of the most commonly used technical indicators.

\medskip
Two future research directions seem to us most natural. The first is to study how the dependence of the optimal portfolio on price history changes in the presence of transaction costs. The second is to move from one investor to many. If all investors in an economy face the same latent drift model, and their portfolio choices are aggregated at equilibrium in a partially observable mean field game, one can ask whether moving-average rules emerge endogenously in the price process itself. They would then be not only an optimal response to prices, but a property of them.

\clearpage
{\emergencystretch=1.5em \printbibliography}
 
\clearpage
\appendix
\numberwithin{equation}{section}
 
\section{HJB and Riccati Calculations}\label{hjb_riccati_appendix}

\subsection{Proof of Proposition \ref{power_coefficient_solvability}}\label{appendix_power_riccati_estimates}

\begin{proof}
Fix $T>0$. The filtering equations imply that $\volatilitystate_t$ is continuous on $\T$, and hence bounded. Reverse time by defining
\begin{equation*}
\widetilde{\boldsymbol v}_s:=\boldsymbol v_{T-s},
\qquad
\alpha:=\frac{1-q}{q\sigma_p},
\qquad
A_s:=\mathcal K+\alpha\widetilde{\boldsymbol v}_s\mathfrak c^\intercal,
\qquad
R_s:=\frac{2}{q}\widetilde{\boldsymbol v}_s\widetilde{\boldsymbol v}_s^\intercal,
\qquad
D:=\frac{q-1}{2q\sigma_p^2}\mathfrak c\mathfrak c^\intercal. 
\end{equation*}
Consider the initial value problem
\begin{equation}\label{reversed_riccati_eq}
\boldsymbol S_s'
=A_s^\intercal\boldsymbol S_s+\boldsymbol S_sA_s
-\boldsymbol S_sR_s\boldsymbol S_s+D,
\qquad
\boldsymbol S_0=0.
\end{equation}
Define $F(s, M):=A_s^\intercal M+MA_s-MR_sM+D$, so that $S_s'=F(s, S_s)$. Since $A_s$ and $R_s$ are continuous and bounded, and $F$ is polynomial in the entries of $M$, it is locally Lipschitz in $M$. Hence, the initial value problem \eqref{reversed_riccati_eq} has a unique local $C^1$ solution.
Continue this solution for as long as possible within $\T$, and write $[0,\tau)$ for the resulting interval.
 
\medskip
Since $R_s$ and $D$ are positive semidefinite, the uniqueness of the solution to \eqref{reversed_riccati_eq} implies that $\boldsymbol S_s$ is symmetric. Towards proving that $S_s$ is positive semidefinite, fix \(s<\tau\) and \(x\in\mathbb R^3\). Let $z:[0,s]\to\R^3$ be the unique solution to the linear ODE
\begin{equation*}
z'(r)=\left(-A_r+\frac12R_r\boldsymbol S_r\right)z(r),
\qquad z(s)=x.
\end{equation*}
Then,
\begin{align*}
x^\intercal\boldsymbol S_sx&=
z(0)^\intercal S_0 z(0)+\int_0^s\frac{d}{dr}\bigl(z(r)^\intercal\boldsymbol S_rz(r)\bigr)\dd r\\
&=\int_0^sz(r)^\intercal\left(
\boldsymbol S_r'
-A_r^\intercal\boldsymbol S_r
-\boldsymbol S_rA_r
+\boldsymbol S_rR_r\boldsymbol S_r
\right)z(r)\dd r\\
&=\int_0^s z(r)^\intercal Dz(r)\dd r\geq0.
\end{align*}
It follows that $\boldsymbol S_s$ is positive semidefinite on $[0,\tau)$.
 
\medskip
Let $\|\Cdot\|$ denote the matrix operator norm, and set
\begin{equation*}
a:=\max_{s\in\T}\|A_s\|,
\qquad
\delta:=\operatorname{tr}(D),
\qquad
m(s):=\operatorname{tr}(\boldsymbol S_s).
\end{equation*}
Positive semidefiniteness of $\boldsymbol S_s$ implies that $0\leq\|\boldsymbol S_s\|\leq m(s)$. Taking traces in \eqref{reversed_riccati_eq} gives
\begin{align*}
m'(s)&=\operatorname{tr}(A_s^\intercal\boldsymbol S_s+\boldsymbol S_sA_s)-\operatorname{tr}(\boldsymbol S_sR_s\boldsymbol S_s)+\delta\\
&\leq\operatorname{tr}(A_s^\intercal\boldsymbol S_s+\boldsymbol S_sA_s)+\delta\\
&\leq 2ah(s)+\delta.
\end{align*}
Since $m(0)=0$, Gronwall's inequality gives
\begin{equation*}
m(s)\leq\delta e^{2as}\int_0^s e^{-2ar}\,dr\leq\delta s e^{2as}.
\end{equation*}
In particular,
\begin{equation*}
\|\boldsymbol S_s\|\leq B:=\delta T e^{2aT},\qquad 0\leq s<\tau.
\end{equation*}
 
Equation \eqref{reversed_riccati_eq} and the preceding bound imply
\begin{equation*}
\|\boldsymbol S_s'\|\leq2aB+\max_{0\leq s\leq T}\|R_s\|B^2+\|D\|,\qquad 0\leq s<\tau,
\end{equation*}
Therefore, for $0\leq u<v<\tau$,
\begin{equation*}
\|\boldsymbol S_v-\boldsymbol S_u\|\leq C(v-u).
\end{equation*}
It follows that $\boldsymbol S_s$ has a continuous extension $\boldsymbol S_\tau$ as $s\uparrow\tau$.
 
\medskip
{\emergencystretch=1.5em If $\tau<T$, the same local existence theorem applied at $(\tau,\boldsymbol S_\tau)$ gives a solution on an interval to the right of $\tau$, contradicting $[0,\tau)$ being the longest possible continuation within $\T$. Hence we must have $\tau=T$, and so
\begin{equation*}
\boldsymbol S\in C^1(\T;\mathrm{Sym}(3,\mathbb R)).
\end{equation*}
Local uniqueness on overlapping intervals gives uniqueness over $\T$. One can easily check that $\quadraticcoef_t:=-\boldsymbol S_{T-t}$ is a $C^1(\T;\mathrm{Sym}(3,\R))$ solution to the first equation in the system \eqref{power_reduced_ode}, and is unique by the global uniqueness of $\boldsymbol S$ as a solution of \eqref{reversed_riccati_eq}.\par}
 
\medskip
With $\quadraticcoef$ known, the equation for $\linearcoef$ is linear with continuous coefficients, and thus has a unique $C^1(\T;\R^3)$ solution. Finally with $\linearcoef$ also known, $f$ is uniquely determined by integration.
\end{proof}

\subsection{Proof of Proposition \ref{hjb_reduction}}\label{appendix_hjb_reduction}

\begin{proof}
We treat the three utilities in turn. In each case, the displayed derivatives of $\psi_U$ give $\partial_{\cwealth\cwealth}\psi_U<0$ and the risk-tolerance identity, so the supremum in \eqref{generic_hjb} is attained at the strategy \eqref{candidate_strategy}.

\medskip
\noindent\emph{Logarithmic utility.} For $\psi_{\log}$ in \eqref{log_candidate_value}, we compute
\begin{equation*}
    \partial_\cwealth\psi_{\log}=\frac{1}{\cwealth},\qquad
    \partial_{\cwealth\cwealth}\psi_{\log}=-\frac{1}{\cwealth^2},\qquad
    \partial_{\cwealth\cprice}\psi_{\log}=0,\qquad
    \nabla\partial_\cwealth\psi_{\log}=0.
\end{equation*}
Hence \eqref{candidate_strategy} reduces to \eqref{log_candidate_strategy}. Substituting this optimizer into the HJB equation gives
\begin{align*}
    0&=\partial_t\psi_{\log}+\mathfrak L\psi_{\log}\!\mid_{\varphi=\varphi_{\log}^*}\\
    &=\partial_t\quadraticstate+\frac{(\lambda_p\cfast-\kappa_p\cprice)^2}{2\sigma_p^2}
    +(\lambda_p\cfast-\kappa_p\cprice)\,\partial_\cprice\quadraticstate
    +(\mu-\kappa \cdrift)\Cdot\nabla\quadraticstate\\
    &\quad+\frac12\sigma_p^2\,\partial_{\cprice\cprice}\quadraticstate
    +\frac12\mathrm{Tr}\!\left(K_tK_t^{\intercal}D^2\quadraticstate\right)
    +\sigma_p K_t\Cdot\nabla\partial_\cprice\quadraticstate.
\end{align*}
Since $\lambda_p\cfast-\kappa_p\cprice=\const^{\intercal}\cpricedrift$, this is
\begin{equation}\label{log_reduced_hjb}
    0=\partial_t\quadraticstate
    +(\K\cpricedrift+\augdrift)\Cdot\nabla_{\cpricedrift}\quadraticstate
    +\frac12\,\volatilitystate_t^{\intercal}(D_{\cpricedrift}^2\quadraticstate)\volatilitystate_t
    +\frac{1}{2\sigma_p^2}\bigl(\const^{\intercal}\cpricedrift\bigr)^2.
\end{equation}
Substituting the quadratic form of $\quadraticstate$ and rearranging
terms yields
\begin{align*}
    0&=\cpricedrift^{\intercal}
    \left(
        \dot{\quadraticcoef}_t
        +\quadraticcoef_t\K
        +\K^{\intercal}\quadraticcoef_t
        +\frac{1}{2\sigma_p^2}\const\const^{\intercal}
    \right)
    \cpricedrift \\
    &\quad+\left(
        \dot{\linearcoef}_t
        +\K^{\intercal}\linearcoef_t
        +2\quadraticcoef_t\augdrift
    \right)\Cdot\cpricedrift \\
    &\quad
    +\dot{\quadraticshift}(t)
    +\augdrift\Cdot\linearcoef_t
    +\volatilitystate_t^{\intercal}\quadraticcoef_t\volatilitystate_t.
\end{align*}
This identity holds for every $\cpricedrift\in\R^3$ if and only if its quadratic, linear, and constant parts vanish separately, which is the system \eqref{log_reduced_ode}. The terminal conditions of \eqref{log_reduced_ode} give $\quadraticstate(T,\cpricedrift)=0$, so that $\psi_{\log}$ satisfies the terminal condition in \eqref{generic_hjb}.

\medskip
\noindent\emph{Power utility.} For $\psi_{\mathrm{pow}}$ in \eqref{power_candidate_value}, we compute
\begin{equation*}
    \partial_\cwealth\psi_{\mathrm{pow}}=\cwealth^{-q}\exp\!\bigl(\quadraticstate(t,\cpricedrift)\bigr),\qquad
    \partial_{\cwealth\cwealth}\psi_{\mathrm{pow}}=-q\,\cwealth^{-q-1}\exp\!\bigl(\quadraticstate(t,\cpricedrift)\bigr),
\end{equation*}
\begin{equation*}
    \nabla_{\mathfrak{z}}\partial_{\cwealth}\psi_{\mathrm{pow}}=\cwealth^{-q}\exp\!\bigl(\quadraticstate(t,\cpricedrift)\bigr)(2\quadraticcoef_t\mathfrak{z}+\linearcoef_t).
\end{equation*}
Hence \eqref{candidate_strategy} reduces to \eqref{power_candidate_strategy}. Substituting this optimizer into the HJB equation, dividing by $\psi_{\mathrm{pow}}$, and using
$\lambda_p\cfast-\kappa_p\cprice=\const^{\intercal}\cpricedrift$ gives
\begin{multline}\label{power_reduced_hjb}
    0=\partial_t\quadraticstate+(\K \cpricedrift+\augdrift)\Cdot\nabla_{\cpricedrift}\quadraticstate
    +\frac{1}{2}\volatilitystate_t^{\intercal}\left(D^2_{\cpricedrift}\quadraticstate+(\nabla_{\cpricedrift}\quadraticstate)(\nabla_{\cpricedrift}\quadraticstate)^{\intercal}\right)\volatilitystate_t\\
    +\frac{1-q}{2q}\left(\frac{\const^{\intercal}\cpricedrift}{\sigma_p}+\volatilitystate_t^{\intercal}\nabla_{\cpricedrift}\quadraticstate\right)^2.
\end{multline}
Using
\begin{equation*}
    \nabla_{\cpricedrift}\quadraticstate=2\quadraticcoef_t\cpricedrift+\linearcoef_t,
    \qquad
    D^2_{\cpricedrift}\quadraticstate=2\quadraticcoef_t,
\end{equation*}
expanding the square and collecting quadratic, linear, and constant terms
yields
\begin{align*}
    0&=\cpricedrift^{\intercal}
    \Biggl(
        \dot{\quadraticcoef}_t
        +\quadraticcoef_t\K+\K^{\intercal}\quadraticcoef_t
        +\frac{2}{q}\quadraticcoef_t\volatilitystate_t\volatilitystate_t^{\intercal}\quadraticcoef_t+\frac{1-q}{q\sigma_p}
        \bigl(
            \const\volatilitystate_t^{\intercal}\quadraticcoef_t
            +\quadraticcoef_t\volatilitystate_t\const^{\intercal}
        \bigr)
        +\frac{1-q}{2q\sigma_p^2}\const\const^{\intercal}
    \Biggr)\cpricedrift\\
    &\quad
    +\Biggl(
        \dot{\linearcoef}_t
        +\K^{\intercal}\linearcoef_t
        +2\quadraticcoef_t\augdrift
        +\frac{2}{q}(\volatilitystate_t^{\intercal}\linearcoef_t)\,\quadraticcoef_t\volatilitystate_t
        +\frac{1-q}{q\sigma_p}(\volatilitystate_t^{\intercal}\linearcoef_t)\,\const
    \Biggr)\Cdot\cpricedrift\\
    &\quad
    +\dot{\quadraticshift}(t)
    +\augdrift\Cdot\linearcoef_t
    +\volatilitystate_t^{\intercal}\quadraticcoef_t\volatilitystate_t
    +\frac{1}{2q}(\volatilitystate_t^{\intercal}\linearcoef_t)^2.
\end{align*}
This identity holds for every $\cpricedrift\in\R^3$ if and only if its quadratic, linear, and constant parts vanish separately, which is the system \eqref{power_reduced_ode}. The terminal conditions of \eqref{power_reduced_ode} give $\quadraticstate(T,\cpricedrift)=0$, so that $\psi_{\mathrm{pow}}$ satisfies the terminal condition in \eqref{generic_hjb}.

\medskip
\noindent\emph{Exponential utility.} For $\psi_{\exp}$ in \eqref{exponential_candidate_value}, we compute
\begin{equation*}
    \partial_\cwealth\psi_{\exp}=-p\,\psi_{\exp},
    \qquad
    \partial_{\cwealth\cwealth}\psi_{\exp}=p^2\,\psi_{\exp},
    \qquad
    \nabla_{\mathfrak z}\partial_{\cwealth}\psi_{\exp}=-p\psi_{\exp}\nabla_{\mathfrak z}\quadraticstate.
\end{equation*}
Hence \eqref{candidate_strategy} reduces to \eqref{exponential_candidate_strategy}. Substituting this optimizer into the HJB equation, dividing by $\psi_{\exp}$, and using
$\lambda_p\cfast-\kappa_p\cprice=\const^{\intercal}\cpricedrift$ gives
\begin{equation*}
    0=\partial_t\quadraticstate
    +(\K\cpricedrift+\augdrift)\Cdot\nabla_{\cpricedrift}\quadraticstate
    +\frac12\volatilitystate_t^{\intercal}
    \left(
        D_{\cpricedrift}^2\quadraticstate
        +(\nabla_{\cpricedrift}\quadraticstate)(\nabla_{\cpricedrift}\quadraticstate)^{\intercal}
    \right)
    \volatilitystate_t
    -\frac{1}{2\sigma_p^2}
    \left(
        \const^{\intercal}\cpricedrift
        +\sigma_p\volatilitystate_t^{\intercal}\nabla_{\cpricedrift}\quadraticstate
    \right)^2.
\end{equation*}
Expanding the square, the terms in $(\volatilitystate_t^{\intercal}\nabla_{\cpricedrift}\quadraticstate)^2$ cancel, and
\begin{equation}\label{exp_reduced_hjb}
    0=\partial_t\quadraticstate
    +(\K\cpricedrift+\augdrift)\Cdot\nabla_{\cpricedrift}\quadraticstate
    +\frac12\volatilitystate_t^{\intercal}(D_{\cpricedrift}^2\quadraticstate)\volatilitystate_t
    -\frac{1}{\sigma_p}(\const^{\intercal}\cpricedrift)\,
    \volatilitystate_t^{\intercal}\nabla_{\cpricedrift}\quadraticstate
    -\frac{1}{2\sigma_p^2}(\const^{\intercal}\cpricedrift)^2.
\end{equation}
Substituting the quadratic ansatz and rearranging gives
\begin{align*}
    0
    &=\cpricedrift^{\intercal}
    \Biggl(
        \dot{\quadraticcoef}_t
        +\quadraticcoef_t\K+\K^{\intercal}\quadraticcoef_t
        -\frac{1}{\sigma_p}
        \bigl(
            \const\volatilitystate_t^{\intercal}\quadraticcoef_t
            +\quadraticcoef_t\volatilitystate_t\const^{\intercal}
        \bigr)
        -\frac{1}{2\sigma_p^2}\const\const^{\intercal}
    \Biggr)\cpricedrift\\
    &\quad
    +\Biggl(
        \dot{\linearcoef}_t
        +\K^{\intercal}\linearcoef_t
        +2\quadraticcoef_t\augdrift
        -\frac{1}{\sigma_p}(\volatilitystate_t^{\intercal}\linearcoef_t)\,\const
    \Biggr)\Cdot\cpricedrift\\
    &\quad
    +\dot{\quadraticshift}(t)
    +\augdrift\Cdot\linearcoef_t
    +\volatilitystate_t^{\intercal}\quadraticcoef_t\volatilitystate_t.
\end{align*}
This identity holds for every $\cpricedrift\in\R^3$ if and only if its quadratic, linear, and constant parts vanish separately, which is the system \eqref{exp_reduced_ode}. The terminal conditions of \eqref{exp_reduced_ode} give $\quadraticstate(T,\cpricedrift)=0$, so that $\psi_{\exp}$ satisfies the terminal condition in \eqref{generic_hjb}.
\end{proof}

\section{Proof of the Verification Theorem}
\label{admissibility_verification_appendix}

We prove Theorem \ref{verification_theorem} for the problem started at time zero from $(v_0,p_0,\kdrift_0)$. The proof for the problem started at time $t$ from a state in $\mathcal O_U$ is identical, with $[t,T]$ in place of $\T$ and the density \eqref{pricing_density} started at time $t$. The proof has three steps: the budget constraint $\E^{\mathbb Q}[V_T^\varphi]\leq v_0$ for every admissible strategy (Lemma \ref{pricing_measure_budget_inequality}), the proportionality of the candidate strategy's marginal utility $U'(V_T^*)$ to the pricing density $Z_T$ (Proposition \ref{candidate_strategy_properties}), and the concavity of $U$. Throughout, $\cpricedrift_t:=(P_t,\kdrift_t)^{\intercal}$ denotes the price--drift state process. By \eqref{controlled_state_dynamics},
\begin{equation}\label{price_drift_state_dynamics}
    \dd\cpricedrift_t=(\K\cpricedrift_t+\augdrift)\dd t+\volatilitystate_t\dd\nu_t,
    \qquad
    \cpricedrift_0=(p_0,\kdrift_0)^{\intercal},
\end{equation}
and $\pi_t=\const^{\intercal}\cpricedrift_t$, $\theta_t=\const^{\intercal}\cpricedrift_t/\sigma_p$.

\subsection{Affine-Gaussian exponential martingales}
\label{appendix_affine_gaussian_martingale}

\begin{lemma}\label{affine_exponential_martingale_lemma}
Let $W$ be a one-dimensional Brownian motion, and let $X$ be an $\R^d$-valued process satisfying
\begin{equation*}
    \dd X_t=(A_tX_t+a_t)\dd t+\beta_t\dd W_t,\qquad X_0=x_0,
\end{equation*}
where $A:\T\to\R^{d\times d}$, $a:\T\to\R^d$, and $\beta:\T\to\R^d$ are continuous deterministic functions. Let $\xi_t=u_t\Cdot X_t+v_t$ where $u:\T\to\R^d$ and $v:\T\to\R$ are continuous deterministic functions. Then the stochastic exponential
\begin{equation*}
    Z_t:=\exp\left(\int_0^t\xi_s\dd W_s-\frac12\int_0^t\xi_s^2\dd s\right)
\end{equation*}
is a true martingale on $\T$.
\end{lemma}
 
\begin{proof}
Since $X$ has continuous paths on the compact interval $\T$, and $\xi_t=u_t\Cdot X_t+v_t$ has continuous sample paths, we have $\int_0^T \xi_t^2\dd t<\infty$ a.s. Hence $Z$ is a well-defined nonnegative local martingale.
 
\medskip
For each $n\in\N$, define
\begin{equation*}
    \tau_n:=\inf\left\{t\in\T:\int_0^t\xi_s^2\dd s\ge n\right\}\wedge T,\qquad
    Z_t^{(n)}:=\exp\left(\int_0^t\xi_s\boldsymbol{1}_{\{s\le\tau_n\}}\dd W_s-\frac12\int_0^t\xi_s^2 \boldsymbol{1}_{\{s\le\tau_n\}}\dd s\right).
\end{equation*}
Since $\int_0^T\xi_s^2 \boldsymbol{1}_{\{s\le\tau_n\}}\dd s\leq n$ a.s., Novikov's criterion implies that $Z^{(n)}$ is a true martingale on $\T$, and therefore $\E[Z_T^{(n)}]=1$.
 
\medskip
Define a probability measure $\prob^{(n)}$ on $\F_T$ by $\frac{\dd\prob^{(n)}}{\dd\prob}:=Z_T^{(n)}$. By Girsanov's theorem,
\begin{equation*}
    W_t^{(n)}:=W_t-\int_0^t\xi_s\boldsymbol{1}_{\{s\leq\tau_n\}}\dd s
\end{equation*}
is a Brownian motion under $\prob^{(n)}$. Under $\prob^{(n)}$, the process $X$ satisfies
\begin{align*}
    \dd X_t&=(A_tX_t+a_t)\dd t+\beta_t\left(\dd W_t^{(n)}+\xi_t\boldsymbol{1}_{\{t\leq\tau_n\}}\dd t\right)\\
    &=\Bigl(A_tX_t+a_t+\beta_t(u_t\Cdot X_t+v_t)\boldsymbol{1}_{\{t\leq\tau_n\}}\Bigr)\dd t+\beta_t\dd W_t^{(n)}.
\end{align*}
Since all coefficients are continuous and deterministic on $\T$, there exists $C>0$ such that
\begin{equation*}
    \left|A_tx+a_t+\beta_t(u_t\Cdot x+v_t)\boldsymbol{1}_{\{t\leq\tau_n\}}\right|+|\beta_t|\leq C(1+|x|)
\end{equation*}
for all $(t,x)\in\T\times\R^d$, uniformly in $n$. Let $\tau_R$ be the first exit time of $X$ from the ball of radius $R$. Applying It\^o's formula to $|X|^2$ under $\prob^{(n)}$ up to $\tau_R$, taking expectations, and using Young's inequality, we obtain
\begin{equation*}
    \E^{(n)}|X_{t\wedge\tau_R}|^2\leq|x_0|^2+C\int_0^t\Bigl(1+\E^{(n)}|X_{s\wedge\tau_R}|^2\Bigr)\dd s,
\end{equation*}
with $C$ independent of $n$ and $R$. Gronwall's inequality, followed by Fatou's lemma as $R\to\infty$, therefore gives $\sup_{n\in\N}\sup_{t\in\T}\E^{(n)}|X_t|^2<\infty$. Since $\xi_t=u_t\Cdot X_t+v_t$, there exists $C>0$ such that $\xi_t^2\leq C(1+|X_t|^2)$ for all $t\in\T$.
Hence
\begin{equation*}
    \sup_{n\in\N}\E^{(n)}\!\left[\int_0^T\xi_s^2 \boldsymbol{1}_{\{s\leq\tau_n\}}\dd s
    \right]=
    \sup_{n\in\N}\E^{(n)}\!\left[\int_0^{T\wedge\tau_n}\xi_s^2\dd s\right]\leq
    C\int_0^T\left(1+\sup_{n\in\N}\E^{(n)}|X_s|^2\right)\dd s<\infty.
\end{equation*}
Thus the stochastic integral $\int_0^{\Cdot} \xi_s\boldsymbol{1}_{\{s\leq\tau_n\}}\dd W_s^{(n)}$ is a true martingale under $\prob^{(n)}$, and hence has zero expectation at time $T$. Consequently,
\begin{equation*}
    \E\!\left[Z_T^{(n)}\log Z_T^{(n)}\right]=
    \E^{(n)}\!\left[\log Z_T^{(n)}\right]=\frac12\E^{(n)}\!\left[\int_0^{T\wedge\tau_n}\xi_s^2\dd s\right].
\end{equation*}
It follows that
\begin{equation*}
    \sup_{n\in\N}\E\!\left[Z_T^{(n)}\log Z_T^{(n)}\right]<\infty.
\end{equation*}
By de la Vallée-Poussin's criterion, the family $\{Z_T^{(n)}\}_{n\in\N}$ is uniformly integrable. Finally, since $\tau_n\uparrow T$ a.s., we have $Z_T^{(n)}\to Z_T$ a.s. Uniform integrability therefore yields
\begin{equation*}
    \E[Z_T]=\lim_{n\to\infty}\E[Z_T^{(n)}]=1.
\end{equation*}
Since $Z$ is a nonnegative local martingale, it is a supermartingale. Therefore the identity $\E[Z_T]=1$ implies that $Z$ is a true martingale on $\T$.
\end{proof}

Applied with $W=\nu$, $X=\cpricedrift$, and $\xi=-\theta$, Lemma \ref{affine_exponential_martingale_lemma} shows that the density $Z$ in \eqref{pricing_density} is a martingale, as claimed in Section \ref{admissible_strategies_subsection}. Under the pricing measure $\mathbb Q$, \eqref{price_drift_state_dynamics} becomes
\begin{equation}\label{observed_market_q_dynamics}
    \dd\cpricedrift_t
    =\left[
        \left(
            \K-\frac{1}{\sigma_p}
            \volatilitystate_t\const^{\intercal}
        \right)\cpricedrift_t+\augdrift
      \right]\dd t
      +\volatilitystate_t\dd\nu_t^{\mathbb Q},
\end{equation}
so $\cpricedrift$ is a linear Gaussian diffusion under $\mathbb Q$ as well, and Lemma \ref{affine_exponential_martingale_lemma} also applies under $\mathbb Q$ with $W=\nu^{\mathbb Q}$.

\subsection{The budget constraint}
\label{appendix_budget_constraint}

\begin{lemma}[Budget constraint]\label{pricing_measure_budget_inequality}
For every $\varphi\in\mathcal A_U(\mathbb F^P)$, the terminal wealth $V_T^\varphi$ is integrable under $\mathbb Q$ and
\begin{equation}\label{arbitrary_strategy_budget}
    \E^{\mathbb Q}[V_T^\varphi]\leq v_0.
\end{equation}
\end{lemma}

\begin{proof}
By \eqref{self_financing} and \eqref{pricing_measure_dynamics}, $\dd V_t^\varphi=\sigma_p\varphi_t\dd\nu_t^{\mathbb Q}$, so $V^\varphi$ is a local martingale under $\mathbb Q$. Under logarithmic and power utility it is positive, and a positive local martingale is a supermartingale. Under exponential utility it is a supermartingale by Definition \ref{admissible_strategy}. In both cases $V_T^\varphi$ is integrable under $\mathbb Q$ and \eqref{arbitrary_strategy_budget} holds.
\end{proof}

\subsection{The candidate strategies}
\label{appendix_candidate_strategies}

Let $(\quadraticcoef,\linearcoef,\quadraticshift)$ be the coefficients of the utility under consideration, and write $\quadraticstate_t:=\quadraticstate(t,\cpricedrift_t)$. By It\^o's formula and \eqref{price_drift_state_dynamics},
\begin{equation}\label{quadratic_state_ito}
    \dd\quadraticstate_t=\mathcal L^0\quadraticstate(t,\cpricedrift_t)\dd t+\Gamma_t\dd\nu_t,
    \qquad
    \Gamma_t:=\volatilitystate_t^{\intercal}\bigl(2\quadraticcoef_t\cpricedrift_t+\linearcoef_t\bigr),
\end{equation}
where $\mathcal L^0\quadraticstate:=\partial_t\quadraticstate+(\K\cpricedrift+\augdrift)\Cdot\nabla_{\cpricedrift}\quadraticstate+\frac12\volatilitystate_t^{\intercal}(D^2_{\cpricedrift}\quadraticstate)\volatilitystate_t$. In this notation, the reduced HJB equations \eqref{log_reduced_hjb}, \eqref{power_reduced_hjb}, and \eqref{exp_reduced_hjb}, evaluated along the state process, read
\begin{equation}\label{reduced_hjb_along_state}
    \mathcal L^0\quadraticstate(t,\cpricedrift_t)=
    \begin{cases}
        -\dfrac12\theta_t^2,&U=U_{\log},\\[8pt]
        -\dfrac12\Gamma_t^2-\dfrac{1-q}{2q}(\theta_t+\Gamma_t)^2,&U=U_{\mathrm{pow}},\\[8pt]
        \theta_t\Gamma_t+\dfrac12\theta_t^2,&U=U_{\exp}.
    \end{cases}
\end{equation}
The processes $\theta$ and $\Gamma$ are affine in $\cpricedrift$ with continuous deterministic coefficients, and $\cpricedrift$ is a linear Gaussian diffusion under both $\prob$ and $\mathbb Q$. Hence
\begin{equation}\label{affine_square_integrability}
    \E\biggl[\int_0^T\bigl(\theta_t^2+\Gamma_t^2\bigr)\dd t\biggr]<\infty,
    \qquad
    \E^{\mathbb Q}\biggl[\int_0^T\bigl(\theta_t^2+\Gamma_t^2\bigr)\dd t\biggr]<\infty.
\end{equation}

\begin{proposition}[Candidate strategies]\label{candidate_strategy_properties}
For each $U\in\{U_{\log},U_{\mathrm{pow}},U_{\exp}\}$, the wealth equation
\begin{equation*}
    \dd V_t^*=\varphi_U^*(t,V_t^*,\cpricedrift_t)\dd P_t,\qquad V_0^*=v_0,
\end{equation*}
has a unique solution, and the strategy $\varphi_t^*:=\varphi_U^*(t,V_t^*,\cpricedrift_t)$ belongs to $\mathcal A_U(\mathbb F^P)$. Moreover, with $y:=\partial_{\cwealth}\psi_U(0,v_0,p_0,\kdrift_0)>0$,
\begin{itemize}
    \item[(i)] $\E^{\mathbb Q}[V_T^*]=v_0$;
    \item[(ii)] $U'(V_T^*)=yZ_T$;
    \item[(iii)] $U(V_T^*)$ is integrable and $\E[U(V_T^*)]=\psi_U(0,v_0,p_0,\kdrift_0)$.
\end{itemize}
\end{proposition}

\begin{proof}
By \eqref{observed_price_dynamics}, $\dd P_t=\sigma_p(\theta_t\dd t+\dd\nu_t)=\sigma_p\dd\nu_t^{\mathbb Q}$. We treat the three utilities in turn.

\medskip
\noindent\emph{Logarithmic utility.} By \eqref{log_candidate_strategy}, the candidate strategy is $V_t^*\theta_t/\sigma_p$, so that $\dd V_t^*=V_t^*\theta_t(\theta_t\dd t+\dd\nu_t)$. The unique solution is
\begin{equation*}
    V_t^*=v_0\exp\left(\int_0^t\theta_s\dd\nu_s+\frac12\int_0^t\theta_s^2\dd s\right)=\frac{v_0}{Z_t}.
\end{equation*}
It is positive, so the candidate strategy is admissible, and $\E^{\mathbb Q}[V_T^*]=\E[Z_TV_T^*]=v_0$, which is (i). Since $y=1/v_0$, we have $U_{\log}'(V_T^*)=1/V_T^*=yZ_T$, which is (ii). By \eqref{quadratic_state_ito}, \eqref{reduced_hjb_along_state}, and $\quadraticstate_T=0$,
\begin{equation*}
    \frac12\int_0^T\theta_t^2\dd t=\quadraticstate_0+\int_0^T\Gamma_t\dd\nu_t,
\end{equation*}
and therefore
\begin{equation*}
    \log V_T^*=\log v_0+\quadraticstate_0+\int_0^T(\theta_t+\Gamma_t)\dd\nu_t
    =\psi_{\log}(0,v_0,p_0,\kdrift_0)+\int_0^T(\theta_t+\Gamma_t)\dd\nu_t.
\end{equation*}
By \eqref{affine_square_integrability}, the stochastic integral is square-integrable with mean zero, which gives (iii).

\medskip
\noindent\emph{Power utility.} By \eqref{power_candidate_strategy}, the candidate strategy is $V_t^*(\theta_t+\Gamma_t)/(q\sigma_p)$, so that
\begin{equation*}
    \frac{\dd V_t^*}{V_t^*}=\frac{\theta_t+\Gamma_t}{q}\dd\nu_t^{\mathbb Q}.
\end{equation*}
The unique solution is a stochastic exponential. It is positive, so the candidate strategy is admissible. Its integrand is affine in $\cpricedrift$, so by Lemma \ref{affine_exponential_martingale_lemma}, applied under $\mathbb Q$, $V^*$ is a $\mathbb Q$-martingale, which gives (i). Let $\Lambda_t:=\partial_{\cwealth}\psi_{\mathrm{pow}}(t,V_t^*,\cpricedrift_t)=(V_t^*)^{-q}e^{\quadraticstate_t}$. By \eqref{quadratic_state_ito} and \eqref{reduced_hjb_along_state},
\begin{align*}
    \dd\log\Lambda_t&=-q\dd\log V_t^*+\dd\quadraticstate_t\\
    &=\left(-(\theta_t+\Gamma_t)\theta_t+\frac{(\theta_t+\Gamma_t)^2}{2q}-\frac12\Gamma_t^2-\frac{1-q}{2q}(\theta_t+\Gamma_t)^2\right)\dd t-\theta_t\dd\nu_t\\
    &=-\frac12\theta_t^2\dd t-\theta_t\dd\nu_t.
\end{align*}
Hence $\Lambda_t=\Lambda_0Z_t=yZ_t$, and since $\quadraticstate_T=0$, $U_{\mathrm{pow}}'(V_T^*)=(V_T^*)^{-q}=\Lambda_T=yZ_T$, which is (ii). Finally, $U_{\mathrm{pow}}(V_T^*)=V_T^*U_{\mathrm{pow}}'(V_T^*)/(1-q)=yZ_TV_T^*/(1-q)$, so by (i),
\begin{equation*}
    \E\bigl[U_{\mathrm{pow}}(V_T^*)\bigr]=\frac{y}{1-q}\E^{\mathbb Q}[V_T^*]=\frac{yv_0}{1-q}=\psi_{\mathrm{pow}}(0,v_0,p_0,\kdrift_0),
\end{equation*}
because $y=v_0^{-q}e^{\quadraticstate_0}$. This is (iii).

\medskip
\noindent\emph{Exponential utility.} By \eqref{exponential_candidate_strategy}, the candidate strategy is $(\theta_t+\Gamma_t)/(p\sigma_p)$, which does not depend on wealth, so that
\begin{equation*}
    \dd V_t^*=\frac{\theta_t+\Gamma_t}{p}\dd\nu_t^{\mathbb Q}.
\end{equation*}
By \eqref{affine_square_integrability}, $V^*$ is a square-integrable $\mathbb Q$-martingale. The candidate strategy is therefore admissible and (i) holds. Let $\Lambda_t:=\partial_{\cwealth}\psi_{\exp}(t,V_t^*,\cpricedrift_t)=p\exp(-pV_t^*+\quadraticstate_t)$. By \eqref{quadratic_state_ito} and \eqref{reduced_hjb_along_state},
\begin{align*}
    \dd\log\Lambda_t&=-p\dd V_t^*+\dd\quadraticstate_t\\
    &=\left(-(\theta_t+\Gamma_t)\theta_t+\theta_t\Gamma_t+\frac12\theta_t^2\right)\dd t-\theta_t\dd\nu_t\\
    &=-\frac12\theta_t^2\dd t-\theta_t\dd\nu_t.
\end{align*}
Hence $\Lambda_t=yZ_t$, and since $\quadraticstate_T=0$, $U_{\exp}'(V_T^*)=pe^{-pV_T^*}=\Lambda_T=yZ_T$, which is (ii). Finally, $U_{\exp}(V_T^*)=-U_{\exp}'(V_T^*)/p=-yZ_T/p$, so
\begin{equation*}
    \E\bigl[U_{\exp}(V_T^*)\bigr]=-\frac yp=\psi_{\exp}(0,v_0,p_0,\kdrift_0),
\end{equation*}
because $y=p\exp(-pv_0+\quadraticstate_0)$. This is (iii).
\end{proof}

\subsection{Proof of Theorem \ref{verification_theorem}}
\label{appendix_verification_proof}

Let $\varphi\in\mathcal A_U(\mathbb F^P)$. Since $U$ is concave and differentiable on its domain, Proposition \ref{candidate_strategy_properties}(ii) gives
\begin{equation}\label{terminal_concavity}
    U(V_T^\varphi)\leq U(V_T^*)+U'(V_T^*)\bigl(V_T^\varphi-V_T^*\bigr)
    =U(V_T^*)+yZ_T\bigl(V_T^\varphi-V_T^*\bigr).
\end{equation}
By Lemma \ref{pricing_measure_budget_inequality} and Proposition \ref{candidate_strategy_properties}(i), $Z_TV_T^\varphi$ and $Z_TV_T^*$ are integrable under $\prob$, and $U(V_T^*)$ is integrable by Proposition \ref{candidate_strategy_properties}(iii). The right-hand side of \eqref{terminal_concavity} is therefore integrable, so the positive part of $U(V_T^\varphi)$ is integrable, and taking expectations gives
\begin{equation*}
    \E\bigl[U(V_T^\varphi)\bigr]
    \leq\E\bigl[U(V_T^*)\bigr]+y\Bigl(\E^{\mathbb Q}[V_T^\varphi]-\E^{\mathbb Q}[V_T^*]\Bigr)
    \leq\psi_U(0,v_0,p_0,\kdrift_0).
\end{equation*}
Taking the supremum over $\varphi$ gives $\val_U(0,v_0,p_0,\kdrift_0)\leq\psi_U(0,v_0,p_0,\kdrift_0)$. By Proposition \ref{candidate_strategy_properties}, the candidate strategy is admissible and attains this bound, so equality holds and the candidate strategy is optimal. The same argument for the problem started at time $t<T$ from any state in $\mathcal O_U$ gives $\psi_U=\val_U$ on $[0,T)\times\mathcal O_U$, and at $t=T$ both sides equal $U(\cwealth)$.

\section{Computational Details for the Numerical Analysis}
\label{computational_appendix}

All calculations use the centered model with $\mu=0$, $p_0=0$, and $\kdrift_0=0$, in which every process has mean zero.

\medskip
\noindent\emph{Stationary moments.} In the stationary state the filter is the steady-state filter \eqref{steady_filter}. Let
\begin{equation*}
    X:=\bigl(P,F,S,\hat F,\hat S,M^{\gamma_1},\dots,M^{\gamma_n}\bigr)^{\intercal},
    \qquad
    \dd M_t^\gamma=\gamma(P_t-M_t^\gamma)\dd t,
\end{equation*}
where the equation for $M^\gamma$ follows from \eqref{price_ema}. By \eqref{standard_filter_form}, \eqref{kalman_bucy_equations}, and \eqref{price_surprise}, $X$ solves a linear equation $\dd X_t=AX_t\dd t+B\dd(W_t^P,{\bf B}_t)$ with constant matrices $A$ and $B$, and its stationary covariance $\Omega$ solves the Lyapunov equation $A\Omega+\Omega A^{\intercal}+BB^{\intercal}=0$. The stationary growth rate is $\E[(\lambda_p\hat F-\kappa_pP)^2]/(2\sigma_p^2)$, and that of the investor who observes the factors replaces $\hat F$ by $F$. For given rates, the best coefficients in each signal class of Section \ref{sensitivity_analysis_section} are those of the projection of $\hat F$ on the regressors, and the loss is $\lambda_p^2/(2\sigma_p^2)$ times the residual variance. The rates are then found by a bounded scalar search for one EMA and by Nelder--Mead searches from thirteen starting pairs for two.

\medskip
\noindent\emph{Economies in Figures \ref{regime_figure} and \ref{value_figure}(a).} With $\sigma$ diagonal, the stationary variance of $F$ is
\begin{equation*}
    \frac{\sigma_f^2}{2\kappa_f}+\frac{\lambda_f^2\sigma_s^2}{2\kappa_f\kappa_s(\kappa_f+\kappa_s)},
\end{equation*}
where the second term is due to the slow factor. For each value of $\kappa_f/\kappa_s$ and of the share of the second term, we set $\lambda_f=\kappa_f$ and choose $\sigma_f$ and $\sigma_s$ so that the sum equals its baseline value, $0.147$. The weak-separation economy has $\kappa_f=\lambda_f=1.2$, $\sigma_f=0.221$, and $\sigma_s=0.478$.

\medskip
\noindent\emph{Finite horizons.} The Riccati equation in \eqref{kalman_bucy_equations} is integrated together with the covariance of $(P,F,S,\hat F,\hat S)$, which gives the growth rate of the optimal portfolio over $\T$, and with the covariance of $\kdrift-\kdrift^\infty$, which by the proof of Theorem \ref{steady_closed_loop_filter} solves $\dot\Upsilon_t=-A_c\Upsilon_t-\Upsilon_tA_c^{\intercal}+(K_t-\bar K)(K_t-\bar K)^{\intercal}$ with $\Upsilon_0=0$. For the trailing window of length $L$, the investor knows the price at the start of the window, so the conditional variance in \eqref{trailing_projection_identity} is the first diagonal entry of the solution at time $L$ of the Riccati equation started from the stationary covariance of $\drift$ given $P$. All differential equations are solved by an adaptive method with relative tolerance $10^{-10}$.

\medskip
\noindent\emph{The myopic portfolio under power and exponential utility.} Repeating the calculations of Appendix \ref{appendix_hjb_reduction} with the myopic portfolio in place of the optimizer shows that its expected utility is $\frac{v_0^{1-q}}{1-q}\exp(\quadraticstate(0,\cpricedrift_0))$ under power utility and $-\exp(-pv_0+\quadraticstate(0,\cpricedrift_0))$ under exponential utility, where the coefficients of $\quadraticstate$ solve the following systems. Under power utility, the system is \eqref{power_reduced_ode} with the terms
\begin{equation*}
    \frac{2}{q}\quadraticcoef_t\volatilitystate_t\volatilitystate_t^{\intercal}\quadraticcoef_t,
    \qquad
    \frac{2}{q}(\volatilitystate_t^{\intercal}\linearcoef_t)\,\quadraticcoef_t\volatilitystate_t,
    \qquad
    \frac{1}{2q}(\volatilitystate_t^{\intercal}\linearcoef_t)^2
\end{equation*}
replaced by
\begin{equation*}
    2\quadraticcoef_t\volatilitystate_t\volatilitystate_t^{\intercal}\quadraticcoef_t,
    \qquad
    2(\volatilitystate_t^{\intercal}\linearcoef_t)\,\quadraticcoef_t\volatilitystate_t,
    \qquad
    \frac12(\volatilitystate_t^{\intercal}\linearcoef_t)^2.
\end{equation*}
Under exponential utility, it is \eqref{exp_reduced_ode} with the latter three terms added to its three equations. Both systems have the sign structure of \eqref{reversed_riccati_eq}, so the argument of Proposition \ref{power_coefficient_solvability} gives a unique solution on $\T$. With $\quadraticstate^*$ and $\quadraticstate^{\mathrm M}$ denoting the functions for the optimal and the myopic portfolio, the share of the certainty-equivalent gain that the myopic portfolio gives up is $1-\quadraticstate^{\mathrm M}(0,\cpricedrift_0)/\quadraticstate^*(0,\cpricedrift_0)$ under both utilities.

\end{document}